\pdfoutput=1
\documentclass[a4paper,12pt,english]{article}

\usepackage[bindingoffset=0.3cm,textheight=22.5cm,hdivide={2.7cm,*,2.7cm}, vdivide={*,22cm,*}]{geometry}
\usepackage{cite}
\usepackage{youngtab}
\usepackage{amsmath,amsfonts,amssymb,babel,slashed,color,empheq,tensor,amsthm}
\usepackage{tensor}
\usepackage{hyperref}

\makeatletter

\newtheorem{thm}{Theorem}[section]
\newtheorem{lem}[thm]{Lemma}

\newcommand{\del}{\partial}

\def\nn{\nonumber} 

\numberwithin{equation}{section}

\def\a{\alpha}  \def\b{\beta}
 \def\g{\gamma} 
 \def\d{\delta} 
    \def\k{\kappa}
   \def\m{\mu}
\def\n{\nu}    \def\r{\rho}
\def\s{\sigma}

  \def\cC{{\cal C}} 
   
 \def\cH{{\cal H}}  
   
\def\cM{{\cal M}} \def\cN{{\cal N}}  
  \def\cR{{\cal R}}

\def\R{{\mathbb R}} \def\C{{\mathbb C}} \def\N{{\mathbb N}}

 \def\one{\mbox{1 \kern-.59em {\rm l}}}

\def\mso{\mathfrak{so}}

\newcommand{\Tr}{\mathrm{Tr}}

\newcommand{\End}{\mathrm{End}}
\def\hs{\mathfrak{hs}}

\def\Tr{\mbox{Tr}}

\newcommand{\eq}[1]{(\ref{#1})}

 \allowdisplaybreaks[3]

\begin{document}



\renewcommand{\title}[1]{\vspace{10mm}\noindent{\Large{\bf#1}}\vspace{8mm}} 
\newcommand{\authors}[1]{\noindent{\large #1}\vspace{5mm}}
\newcommand{\address}[1]{{\itshape #1\vspace{2mm}}}


\begin{titlepage}
\begin{flushright}
 UWThPh-2021-22 \\
 UTHEP-768
\end{flushright}
\begin{center}
\title{ {\Large Spherically symmetric solutions of  
 higher-spin gravity\\[1ex]   in the IKKT matrix model}  }

\vskip 3mm

\authors{Yuhma Asano${}^\dagger$ and Harold C.\ Steinacker${}^\ddagger$}

\vskip 3mm

\address{ 
${}^\dagger${\it Faculty of Pure and Applied Sciences, University of Tsukuba,\\
1-1-1 Tennodai, Tsukuba, Ibaraki 305-8571, Japan}\\
Email: {\tt asano@het.ph.tsukuba.ac.jp}\\
\vskip 3mm
${}^\ddagger${\it Faculty of Physics, University of Vienna\\
Boltzmanngasse 5, A-1090 Vienna, Austria  }  \\
Email: {\tt harold.steinacker@univie.ac.at}  
  }

\bigskip

\vskip 1.4cm

\textbf{Abstract}
\vskip 3mm

\begin{minipage}{14.5cm}%
\vskip 3mm

We present a systematic study of spherically symmetric vacuum solutions of the 
IKKT matrix model, within the framework of semi-classical covariant quantum geometries.
All asymptotically flat solutions of the equations of motion of the frame  are found explicitly.
They reproduce the linearized Schwarzschild geometry for large $r$ but deviate from it at the 
non-linear level, and include contributions from dilaton and axion. 
They are pertinent to the pre-gravity theory arising on classical brane solutions within the classical matrix model, 
before taking into account the Einstein-Hilbert term induced by quantum effects.
We also address the problem of reconstructing matrix configurations corresponding to 
some given frame, and show that this problem can always be solved at the geometrical
level of the underlying higher spin theory, ignoring possible higher spin modes.

\end{minipage}

\end{center}

\end{titlepage}

\tableofcontents
%
%
\section{Introduction}

Matrix models have been introduced some 25 years ago 
as candidates for a non-perturbative formulation of 
superstring theory \cite{Ishibashi:1996xs,Banks:1996vh}. They provide an independent and  non-perturbative starting point, which allows to access the 
rich structures of string theory from a different angle. 
In particular, they yield solutions and configurations
which are not easily seen from the more traditional point of view.
Our approach is to take the IKKT or IIB matrix model as a starting point, 
and investigate the physics emerging 
on interesting background solutions.

A particularly interesting type of 3+1-dimensional covariant quantum space-time
solution of the IKKT model\footnote{See also e.g.~\cite{Castelino:1997rv} 
for a somewhat related early construction in the BFSS model.} was recently found 
\cite{Sperling:2019xar,Steinacker:2017vqw,Steinacker:2017bhb}, which is manifestly invariant under local rotations and translations. 
This type of solution is dynamical, and leads to a well-defined 
higher-spin gauge theory for the fluctuations on the background.
The geometry is governed by a dynamical frame, which is generated 
by the dynamical matrices of the model in the semi-classical regime.
A covariant description of the geometrical sector 
of this theory was obtained in \cite{Steinacker:2020xph} 
in terms of a Weitzenb\"ock connection and its torsion.
This was cast into a more conventional equation in terms of the 
dynamical frame and the Levi-Civita connection in \cite{Fredenhagen:2021bnw}. 
This allows to look for solutions of these non-linear equations 
describing non-trivial geometries. 
A simple spherically symmetric static solution centered at some point 
was  found, which is asymptotically flat and reduces to the 
linearized Schwarzschild solution for large $r$.

In the present paper, we present a systematic study of spherically symmetric vacuum solutions of the non-linear equations of motion for the frame. 
This is a rather non-trivial problem 
due to the presence of dilaton and (gravitational) axion, 
which lead to a coupled system of non-linear equations.
Its solution is not unique, in contrast to general relativity (GR),
 where the Schwarzschild solution is unique by Birkhoff's theorem. 
Despite the complicated structure,
we obtain explicitly 
the most general vacuum solution, 
including non-trivial contributions from dilaton and (gravitational) axion.
The generic solution is given in terms of a hypergeometric function
$_2F_1$ involving a number of free parameters. The solutions with asymptotically flat
geometry involve three independent parameters, which can be identified as
mass, and two scales characterizing the axion and dilaton.
They have an intricate global structure with various  
types of behavior depending on the parameters.

Even though all asymptotically flat solutions reduce to the 
linearized Schwarzschild geometry for large radius,
none of the solutions appear to reproduce the characteristic features of a 
Schwarzschild-like horizon.
The deviation from the 
Schwarzschild geometry is already seen in the 
 Eddington-Robertson-Schiff parameters, which are found to be 
 $\g = 1$ but $\b = 2$ in all asymptotically flat solutions. 
This should not be too surprising, 
since we consider the semi-classical regime of the matrix model without  
taking into account quantum effects.
Interestingly, we find some solutions which are 
reminiscent of wormhole geometries, connecting two asymptotically flat 
geometries.

These non-standard characteristics of the solutions are interpreted as  
features of the classical {\em pre-gravity theory} described by the classical 
matrix model, which is expected to dominate the extreme IR 
(cosmological) regime. In the presence of fuzzy extra dimensions, an Einstein-Hilbert (E-H) term arises
in the quantum effective action as shown in \cite{Steinacker:2021yxt},
 which is expected to dominate on shorter scales.
Therefore the present solutions should be interpreted with caution.
They certainly provide a deeper understanding of the classical 
aspects of the matrix model and its deviations from GR; however a 
proper physical assessment of the solutions can presumably 
only be given once this induced Einstein-Hilbert term is taken into account.

The solutions given in this paper are solutions of the equations of motion for the 
frame, and we also elaborate the corresponding metric in standard form.
To obtain solutions of the (semi-classical) matrix model, it 
remains to be shown that these frames can be implemented 
in terms of Hamiltonian vector fields generated by the basic matrices.
We also address this ``reconstruction'' problem and show that this can always be 
achieved  at the lowest ``geometrical'' level of the 
underlying higher spin gauge theory,
based on general results in \cite{Sperling:2019xar,Sperling:2018xrm}. 
The extension to the full 
higher spin sector remains as an open problem, but we expect that this can be solved,
as illustrated in the simplest solution found in \cite{Fredenhagen:2021bnw}.

This paper is organized as follows. After a brief review of the underlying 
matrix model and its geometrical interpretation in section \ref{sec:review}, 
we discuss the general setup of spherically symmetric geometries in the 
present framework, and obtain
the most general spherically symmetric solutions for the frame 
in section \ref{sec:solution-general}. The effective metric resulting from these 
solutions is then discussed in section \ref{sec:metric}.
In appendix A, we provide a solution of the reconstruction problem 
at the geometrical level of the full higher spin theory.
Finally, appendix B provides a compact derivation of the covariant equations of motion
for the frame and its torsion.

\section{Matrix model and cosmological quantum spacetime}
\label{sec:review}

We consider the IKKT or IIB matrix model \cite{Ishibashi:1996xs} with a mass term,
\begin{equation}
S[Z,\Psi] = {\rm Tr}\big( [Z^{\dot\a},Z^{\dot\b}][Z_{{\dot\a}},Z_{{\dot\b}}] + 2 m^2 Z_{\dot\a} Z^{\dot\a}
\,\, + \overline\Psi \Gamma_{\dot\a}[Z^{\dot\a},\Psi] \big) \ ,
\label{MM-action}
\end{equation} 
where the indices are contracted with the flat metric $\eta_{\dot\a\dot\b}$.
Besides invariance under gauge transformations $Z^{\dot\a}\to U^{-1} Z^{\dot\a} U$,
the model is invariant under $SO(9,1)$ acting on the dotted indices, and $\cN=2$ supersymmetry when $m^2=0$. 
We will study classical solutions of this matrix model,
which can be interpreted as spherically symmetric space-time geometries 
around some center, which for large distances reduce to the cosmological 
space-time solution found in \cite{Sperling:2019xar}.
That solution is given by
\begin{align}
 Z^{\dot \a}  = \frac 1R \cM^{\dot \a 4} =: T^{\dot \a}, \qquad R^{-2} = \frac{1}{3} m^2 \ 
 \label{Y-T-solution}
\end{align}
for ${\dot \a}=0,...,3$,
and  requires the presence of the mass term in \eq{MM-action}. 
Here $\cM^{{\dot \a}{\dot \b}}$ are $\mso(4,2)$ generators in the doubleton 
representation $\cH_n$ for $n\in\N$. 
All matrices with indices from $4$ to $9$ as well as the fermions will be set to zero;
nevertheless, they play a crucial role in the quantum theory.
The mass $m^2$ sets the cosmological curvature scale, 
so that it effectively vanishes from a local point of view. 
Ultimately, it is expected that this solution (or a very similar one)
is stabilized by quantum effects without mass term.
However we restrict ourselves  to the classical model
in this paper, and therefore include the explicit mass term.

%
%
\subsection{Semi-classical structure of the background}
\label{sec:semiclass}

We consider a class of solutions of the 
matrix model, which admit a semi-classical
interpretation in terms of a 6-dimensional
Poisson (more precisely: symplectic)
manifold $\cM^6$, which can be viewed as a twisted $S^2$ bundle over space-time $\cM^{3,1}$:
\begin{align}
 \cM^6 \stackrel{\rm loc}{\cong} \cM^{3,1} \times S^2
 \ .
\end{align}
``Twisted'' bundle\footnote{More precisely: equivariant bundle.}, means that the 
local space-like stabilizer group $SO(3)$ of any point on $\cM^{3,1}$
acts on the internal $S^2$.
Such a background in the matrix model  will be denoted as ``covariant quantum space''.
A specific example of such a solution where $\cM^{3,1}$ is a cosmological
FLRW space-time was given in \cite{Steinacker:2017bhb} and discussed in detail in \cite{Sperling:2019xar}.
The most important feature of such a background in the matrix model is 
that the local bundle structure leads to a higher-spin gauge theory,
where  the  gauge symmetry of the matrix model translates into 
(higher-spin generalization of) volume-preserving diffeomorphism.

It was shown in \cite{Steinacker:2020xph} that the semi-classical equations of motion of the 
matrix model can be translated\footnote{That formulation is justified in the 
asymptotic regime, where the scale of perturbations is much shorter than the cosmic scale.} into a covariant description  of 
a frame and its associated Weitzenb\"ock connection. That description, in turn,
was re-cast in \cite{Fredenhagen:2021bnw} in terms of covariant equations of motion for the frame,
in terms of the standard Levi-Civita connection.

To be specific, we will mostly focus on local perturbations of the cosmological 
solution in \cite{Sperling:2019xar}, keeping the asymptotics fixed. However, 
most of the considerations are more general, and some of the 
solutions will correspond to more general asymptotic backgrounds. 
In particular, we find hints for solutions with global structure 
reminiscent of wormholes, and some special cases are
expected to give rise to other  cosmological asymptotics.
Those aspects should be studied in detail elsewhere.

\paragraph{Cosmological FLRW background.}

Let us describe the mathematical structure of the solution in 
\cite{Sperling:2019xar} in some detail.
The only mathematical structures which exist in the matrix model are the matrices 
$Z^{\dot\a}$ and their commutators, 
which reduce to functions and Poisson brackets in the semi-classical limit $n \to \infty$.
We must learn how to work with these efficiently and to cast the system into a recognizable form.

The above background provides natural generators $x^\mu \sim X^\mu$ and $t^\mu \sim T^\mu$ 
which can be interpreted as functions on $\C P^{1,2}$,  with 
a Poisson or symplectic structure  $\{.,.\}$ encoded in $\theta^{\mu\nu} = \{x^\mu,x^\nu\}$.
The doubleton representations $\cH_n$ entail constraints, which
in the semi-classical limit imply the following relations 
\begin{subequations}
\label{geometry-H-M}
\begin{align}
 x_\mu x^\mu &= -R^2 - x_4^2 = -R^2 \cosh^2(\eta) \, , 
 \qquad R \sim \frac{\tilde R}{2}n   \label{radial-constraint}\\
 t_{\mu} t^{\mu}  &=  \tilde R^{-2}\, \cosh^2(\eta) \,, \\
 t_\mu x^\mu &= 0, \label{xt-constraint} \\
 t_\mu \theta^{\mu\a} &= - \sinh(\eta) x^\a , \\
 x_\mu \theta^{\mu\a} &= - \tilde R^2 R^2 \sinh(\eta) t^\a , \label{x-theta-contract}\\
 \eta_{\mu\nu}\theta^{\mu\a} \theta^{\nu\b} &= R^2 \tilde R^2 \eta^{\a\b} - R^2 \tilde R^4 
t^\a t^\b + \tilde R^2 x^\a x^\b \ ,
%
\end{align}
\end{subequations}
where $\mu,\a = 0,\ldots ,3$, and a self-duality relation for $\theta^{\mu\nu}$ \cite{Sperling:2019xar}.
The $x^\mu$ will be interpreted as 
Cartesian  
coordinate functions $x^\mu:\, \cM^{3,1}\hookrightarrow \R^{3,1}$, and $\eta$ is a global time coordinate via
\begin{align}
 R \sinh(\eta) = \pm \sqrt{-x_\mu x^\mu-R^2} \ ,
 \label{x4-eta-def}
\end{align}
which  will be related to the scale parameter  of the universe. 
Similarly, the $t^\mu$
are extra generators which describe the internal $S^2$ fiber over every point on $\cM \equiv \cM^{3,1}$.
Together, $x^\mu$ and $t^\mu$ generate the algebra $\cC \cong \cC^\infty(\C P^{1,2})$. 
Along with a selfduality relation,
these constraints allow  to express the Poisson tensor $\theta^{\mu\nu}$  as follows
\begin{align}
 \theta^{\mu\nu} &= \frac{\tilde R^2}{\cosh^2(\eta)} 
   \Big(\sinh(\eta) (x^\mu t^\nu - x^\nu t^\mu) +  \epsilon^{\mu\nu\a\b} x_\a t_\b \Big) \ .
   \label{theta-general}
\end{align}
Now consider the Poisson brackets:
 \begin{align}
  \{x^\mu,x^\nu\} &=  \theta^{\mu\nu} = - \tilde R^2 R^2\{t^\mu, t^\nu\}
  \ ,
  \label{X-X-CR}\\
  \{t^\mu,x^\nu\} &= \eta^{\mu\nu} \sinh(\eta) \, .  \label{T-X-CR}   
\end{align}
This implies that
\begin{align}
\{t^\mu,.\}  \ = \sinh(\eta) \del_\mu,
\end{align}
where $\del_\mu:=\frac{\del}{\del x^\mu}$
act as momentum generators on $\cM^{3,1}$,
leading to  the useful relation 
\begin{align}
 \del_\mu \phi = \b\{t_\mu,\phi\}, \qquad \b = \frac{1}{\sinh(\eta)}
 \ ,
 \label{del-t-rel}
\end{align}
for an arbitrary function $\phi = \phi(x)$.
In the late time regime, the internal sphere and the Poisson tensor 
are characterized by \cite{Steinacker:2020xph}
\begin{align}
 |t|   &\approx \tilde R^{-1} \cosh(\eta) \ , \nn\\
 \theta^{0i} &\stackrel{\xi}{\approx} \tilde R^2 R\, t^i  \ \gg \
  \theta^{ij} \stackrel{\xi}{\approx} \frac {\tilde R^2 R}{\sinh(\eta)}  \epsilon^{ijk} t^k  \ \stackrel{\eta\to\infty}{\sim} const \ 
\end{align}
near the reference point $\xi=(x^0,0,0,0)$.

\subsection{Frame and geometry on $\cM^{3,1}$}

The matrix model provides 3+1 generators $Z_{\dot\a}$, which define a frame on 
$\cM^{3,1}$ via 
\begin{align}
  \{Z_{\dot\a},y^\mu\} = \tensor{E}{_{\dot\a}^\mu}
  .
  \label{frame-Poisson}
\end{align}
Here we use general coordinates $y^\mu$ that contain $x^\mu$.
This defines a metric 
\begin{align}
 \g^{\mu\nu} = \eta_{\dot\a\dot\b}\tensor{E}{_{\dot\a}^\mu}\tensor{E}{_{\dot\b}^\nu}
 .
\end{align}
Note that coordinate
indices such as $\mu$ are raised and lowered by $\gamma_{\mu\nu}$,
and frame indices such as $\dot\alpha$  by $\eta_{\dot\alpha\dot\beta}$.
The effective metric $G^{\mu\nu}$ is then given  by a conformal rescaling 
such that \cite{Sperling:2019xar}
\begin{align}
 \sqrt{|G|}\r^{-2} =  \r_M  
 ,
 \label{rho-M-G-relation}
\end{align}
where $|G|$ is the absolute value of the determinant of $G_{\mu\nu}$ and
$\r_M$ is the symplectic volume form (reduced to $\cM^{3,1}$), 
which in Cartesian coordinates is $\r_M =|\theta^{\mu\nu}|^{-\frac{1}{2}}=1/\sinh(\eta)$.
Explicitly,
\begin{align}
 G^{\mu\nu} = \r^{-2} \g^{\mu\nu}
 .
\end{align}
By definition, one can always make $\rho$ positive without loss of generality.
For the cosmic background solution \eq{Y-T-solution}, the frame takes the form
\begin{align}
 \tensor{\bar E}{_{\dot\a}^\mu} =  \sinh(\eta)\d_{\dot\a}^\mu \ ,
 \label{background-frame}
\end{align}
and the dilaton is given by 
\begin{align}
 \bar\rho^2 = \sinh^3(\eta) \ .
 \label{dilaton-BG}
\end{align}
More generally,
the frame resulting from the above construction always satisfies the 
divergence constraint 
\begin{align}
 \del_\nu \big(\r_M\tensor{E}{_{\dot\a}^{\nu}}\big) =
  \del_\nu \big(\sqrt{|G|} \r^{-2}\tensor{E}{_{\dot\a}^{\nu}}\big) 
  = 0 =  \nabla^{(G)}_\nu(\r^{-2}\tensor{E}{_{\dot\a}^{\nu}})
 \ .
 \label{frame-div-free}
\end{align} 
This can be seen as a consequence  of the Jacobi identity 
\cite{Steinacker:2019fcb,Fredenhagen:2021bnw},
and it means that the $\r^{-2}\tensor{E}{_{\dot \a}^\mu}$ are 
volume-preserving vector fields.

\paragraph{Torsion tensor.}
 
 As usual,
we can associate to the  frame $\tensor{E}{_{\dot \a}^\mu}$ the co-frame
\begin{align}
 \tensor{E}{_{\dot \a}_\mu} = \g_{\mu\nu}\tensor{E}{_{\dot \a}^\nu}
 \ ,
\end{align}
which can be viewed as a one-form
\begin{align}
 E_{\dot \a} := \tensor{E}{_{\dot \a}_\mu} dx^\mu
 \ .
\end{align}
Since the frame satisfies a divergence constraint \eq{frame-div-free}, there is 
no local Lorentz invariance acting on the frame index. 
This means that the frame has more physical content than in GR.
In particular, one can define the associated tensor or two-form
\begin{align}
 T^{\dot \a} := d E^{\dot \a} = \frac 12 \tensor{T}{_\mu_\nu^{\dot \a}} dx^\mu dx^\nu\ .
\end{align}
This can be understood as torsion of the Weitzenb\"ock connection 
associated to the frame. 
Its totally antisymmetric component defines a rank one tensor via 
the Hodge star
\begin{align}
  \tilde T_{\r} =  \frac 12\sqrt{|G|}^{-1}\r^2 G_{\r\k}\varepsilon^{\n\s\mu\k}\tensor{T}{_\nu_\s_\mu} \ 
  = \frac 12\sqrt{|\g|}^{-1}\g_{\r\k}\varepsilon^{\n\s\mu\k}\tensor{T}{_\nu_\s_\mu} \ ,
 \label{T-AS-dual-2}
\end{align}
where 
$\tensor{T}{_\nu_\s_\mu} = \tensor{T}{_\nu_\s^{\dot \a}}\tensor{E}{_{\dot \a}_\mu}$.
We shall denote $\tilde T_\mu$ as an axion one-form, for reasons explained in \cite{Fredenhagen:2021bnw}.
Moreover, the contraction of the torsion tensor is related to the dilaton
through the following identity 
\begin{align}
 \frac{2}{\rho}\del_\s\r \ &=  \tensor{T}{_\mu_{\s}^\mu}  \ .
 \label{T-trace-rho}
\end{align}

\paragraph{Equations of motion.}

It was shown in \cite{Fredenhagen:2021bnw}   that 
the semi-classical equations of motion of the matrix model lead to the following
equation of motion for the frame
\begin{align}
 \nabla^{(G)}_\nu(G^{\n\r}\rho^{2}\tensor{T}{_{\r\m}^{\dot\a}})
 = \frac 12 \rho^2 \sqrt{|G|}^{-1}
 \varepsilon^{\n\r\k\s}G_{\m\k} \tilde T_\s \,
 \tensor{T}{_\n_{\r}^{\dot\a}}
 + m^2 \tensor{E}{^{\dot\a}_\m} \ .
 \label{torsion-frame-eom-3}
\end{align}
See appendix \ref{sec:derive-EoM} for details.
These are analogous to Maxwell equations for the 4 vector fields $\tensor{E}{^{\dot\a}_\mu}$.
Moreover, one can show using the equations of motion that
the axion vector field is the derivative of a scalar field
$\tilde\rho$  identified as gravitational axion,
\begin{align}
 \tilde T_\mu = \r^{-2}\del_\mu\tilde\rho \ ,
 \label{axion-def}
\end{align}
which satisfies 
\begin{align}
 \sqrt{|G|} \nabla^{(G)}_\nu (G^{\mu\nu}\r^{-4}\del_\mu\tilde\r)
 = \frac 14  \varepsilon^{\n\r\mu\k}
 \tensor{T}{_\nu_\r^{\dot\a}} \tensor{T}{_\k_\mu_{\dot\a}}  \ .
 \label{div-T-dEdE}
\end{align}
This can be written in terms of differential forms as follows
\begin{align}
 d * (\rho^{-4}d\tilde{\rho}) =  T^{\dot{\alpha}}\wedge T_{\dot{\alpha}}
 \ ,
 \label{div-T-dEdE-form}
\end{align}
where $*$ is the Hodge star associated with the effective metric $G$.
Finally, the dilaton satisfies the following equation of motion 
as part of \eqref{torsion-frame-eom-3}
\begin{align}
 - \nabla^{(G)}_{\nu}\big(G^{\mu\nu}\r^{-1}\del_\mu\r) =
 2 \r^{-2} m^2 + \frac 14\rho^{2}\tensor{T}{_{\mu\r}_\kappa}\tensor{T}{_\nu_\s^\kappa} G^{\mu\nu}G^{\r\s}
 + \frac 1{2} \r^{-4} G^{\mu\nu} \del_\mu\tilde\r \del_\nu\tilde\r \ . 
 \label{Box-rho-onshell}
\end{align}

\section{General rotationally invariant frame}

We are interested in spherically symmetric static geometries centered 
at some point in space, which can be viewed as a local perturbation 
of the cosmic background solution \eq{Y-T-solution}. We  assume that 
the scale of this local 
structure is much smaller than the cosmic background curvature, so that 
 the background frame \eq{background-frame} can be approximated in Cartesian coordinates by 
\begin{align}
 \tensor{E}{^{\dot\a}_\mu} \sim \frac{1}{\sinh(\eta)}\, \d^{\dot\a}_\mu
 \label{background-frame-loc}
\end{align}
for some fixed $\eta$, neglecting the cosmic time evolution. 
The $SO(3)$ symmetry around the local center 
acts on the cosmic background frame by treating $\dot\a$ as a vector index. 
Spherical symmetry will be imposed by keeping this $SO(3)$ symmetry manifest 
also for the perturbed frame. 
This is achieved in Cartesian coordinates centered at $r=0$ via
the  ansatz  
\begin{align}
\tensor{E}{^0_0} &= A(r),\nonumber\\
\tensor{E}{^i_0} &= E(r)\, x^{i},\nonumber\\
\tensor{E}{^0_i} &= D(r)\, x^{i},\nonumber\\
\tensor{E}{^i_j} &= F(r)\, x^{i}x^{j} + \delta^{i}_{\ j}B(r) + S(r)\epsilon_{ijm} x^{m} \ ,
\label{ansatz}
\end{align}
where $A,B,D,E,F$ and $S$ are functions of 
$r^2 = \sum x_i^2$ only.
Such a solution with $S=0$ and $E=0$ was found in \cite{Fredenhagen:2021bnw}, given by
\begin{align}
 \frac{1}{A} &= 1 + \frac {2M}{r} , \qquad
 B_0 = b_0, \qquad  E = 0 = S , \qquad
  A \r^2 = const \ .
 \label{special-solution-HS}
\end{align}
In this paper, we shall find 
the most general spherically symmetric static solution.
$D$ and $F$ can be eliminated  using a simple change of coordinates 
$t\to t +f(r)$ and $x^i \to g(r) x^i$, which is 
understood from now on. 
In terms of differential forms $\tensor{E}{^{\dot\a}}= \tensor{E}{^{\dot\a}_\mu} dx^\mu$,
the frame is then
\begin{align}
\tensor{E}{^0} &= A dt,\qquad 
\tensor{E}{^i} = B dx^i + E x^i dt + S\epsilon_{ijm} x^{m} dx^j
 \ ,
\end{align}
and the associated torsion two-form $T^{\dot\a} = dE^{\dot\a}$ is obtained as
\begin{align}
T^0 &= d E^0 =  A' dr dt, \nn\\
T^i &= d E^i = B' dr dx^i + E d x^i dt + x^i E' dr  dt 
+ S\epsilon_{ijm} d x^{m} dx^j + S' \epsilon_{ijm} x^{m} dr dx^j \ ,
\label{ansatz-E-simple}
\end{align}
where the prime denotes the derivative of functions with respect to $r$.
Then the effective metric  takes the form 
\begin{align}
 G_{\mu\nu}=\r^2
 \begin{pmatrix}
  -(A^2-r^2 E^2) & r B E & 0 & 0 \\
  r B E &  B^2 & 0 & 0 \\
  0 & 0 & r^2 (B^2+r^2 S^2) & 0 \\
  0 & 0 & 0 & r^2 (B^2+r^2 S^2)\sin ^2\vartheta \\
 \end{pmatrix}
 \ ,
\label{eff-etric-general}
\end{align}
in the standard polar coordinates $(t,r,\vartheta,\varphi)$
with 
\begin{align}
 \sqrt{|\g|}
 = |A B| (B^2+r^2 S^2)r^2\sin\vartheta \ .
\end{align}
Throughout this paper, we assume that $A$ and $B$ are positive for $r\to\infty$, 
since we are interested in 
perturbation of \eqref{background-frame-loc}.

\subsection{The general divergence constraint}

We can solve the dilaton constraint \eq{T-trace-rho} in the form
\begin{align*} 
 -\frac{2}{\rho}\partial_{\mu}\rho = \tensor{T}{_\mu_\nu^\nu}
 \ ,
\end{align*}
for the most general ansatz as follows.
Due to the spherical symmetry it suffices to consider the
time and radial components, which reduce
to the following two relations
\begin{align}
0 &= \frac{E}{B}\Big(\Big( \ln \Big|\frac{A}{r E}\Big| \Big)' -\frac{2 B^2}{r \left(B^2+r^2 S^2\right)}\Big) \ , \nn\\
0 &= \left[\ln \big| A \rho ^2 \left(B^2+r^2 S^2\right)\big|\right]'
 +\frac{2 r S^2}{B^2+r^2 S^2} \ .
\label{dilaton-constr-general}
\end{align}
The first equation has two branches: one with $E=0$ and one with $E\neq 0$.


\paragraph{Branch $E = 0$.}
In this case,
the divergence constraint \eq{dilaton-constr-general} reduces to
only one differential equation.
This can be rewritten as
\begin{align}
 (\ln |A|)'
 =\frac{2B^2}{r(B^2 + r^2 S^2)}
 - \left[ \ln\big(r^2 \rho^2 \left(B^2 + r^2 S^2\right)\big) \right]'
 \ .
 \label{dilaton-constr-E=0}
\end{align} 


\paragraph{Branch $E \neq 0$.}

In this case, the difference of the above relations \eqref{dilaton-constr-general} can be integrated as 
\begin{align}
 E = \frac{c_e}{r^3 \rho^2 \left(B^2+r^2 S^2\right)}
 \ ,
 \label{E-B-S-relation}
\end{align}
for some integration constant $c_e$.
Inserting this into the first equation gives 
\begin{align}
 \left(\ln \Big|\frac{A}{r E}\Big| \right)'
 = \frac 2{c_e} r^2 \rho^2 B^2 E \ .
 \label{dilaton-relation-1}
\end{align}
The two equations can  be written in the  equivalent form
\begin{align}
 B^2 &= \frac{c_e }{2r^2 \rho^2}
 \frac{\left(\ln \left|\frac{A}{r E}\right|\right)'}{E}
 \ ,
 \nn\\
 S^2 &= -\frac{c_e}{2r^4 \rho^2}
 \frac{\left(\ln \left|\frac{A}{r^3 E}\right|\right)'}{E}
 \ .
 \label{B-S-AE}
\end{align}
Hence $B(r)$ and $S(r)$ are determined by the 
two arbitrary functions $A(r)$ and $E(r)$. These should be 
determined by the equations of motion.

\subsection{The axion}

The axion one-form $\tilde T = \tilde T_\mu dx^\mu$ \eq{T-AS-dual-2}
can be obtained from 
\begin{equation}
 \frac{1}{2} T_{\nu\rho \mu} dx^{\nu}\wedge dx^{\rho}\wedge dx^{\mu}
 = T^{\dot{\alpha}}\wedge E_{\dot{\alpha}} \ = \r^{-2} *\tilde T 
 \ .
\label{axion-star}
\end{equation} 
By using \eqref{ansatz-E-simple}, 
one can rewrite $T^{\dot\alpha}\wedge T_{\dot\alpha}$ explicitly as
\begin{align}
 T^0 \wedge E_0 &= 0 \ , \nn\\
 T^i \wedge E_i 
 &= (B' dr dx^i +  E d x^i dt ) \wedge S\epsilon_{ijm} x^{m} dx^j
  + S\epsilon_{ijm} d x^{m} dx^j(B dx^i + E x^i dt)
 +  S' B \epsilon_{ijm} x^{m} dr dx^j dx^i \nn\\
 %
 %
 &= 2(r B' S - r S' B - 3SB) d^3 x
  - 2 E S   \epsilon_{mij} x^{m} d x^idx^j dt
 \ ,
\end{align}
noting that 
\begin{align}
 x^i \epsilon_{ijm} rdr\wedge  d x^{j} \wedge dx^m 
  = 2 r^2  d^3 x \ .
  \label{epsdxdx-id}
\end{align}
Through the following relation
\begin{align}
r B' S - r S' B - 3SB
  = r B S \left(\ln\Big|\frac{B}{r^3 S}\Big|\right)' \ ,
\end{align}
the $t$ and $r$ components of \eq{axion-star} reduce to
\begin{align}
 \g^{t\nu}\tilde T_\n
 &= \frac{2 r BS(\ln\left|\frac{B}{r^3 S}\right|)'}
 {|A B|(B^2 + r^2 S^2)} \ , \nn\\
 \g^{r\nu}\tilde T_\n 
 &= \frac{4r E S}{|A B|(B^2 + r^2 S^2)}
 \ ,
 \label{axion-VF-general}
\end{align}
while all other components in $(t,r,\vartheta,\varphi)$ coordinates vanish.
Hence the axion vanishes if $S=0$.
Explicitly, this gives
\begin{align}
 \tilde T_t &= \r^{-2}\del_t\tilde \rho = -\frac{2 r B S}{|A B|(B^2 + r^2 S^2)}
 \left(-r^2 E^2  \left(\ln\Big|\frac{B}{r S}\Big|\right)' + 
       A^2 \left(\ln\Big|\frac{B}{r^3 S}\Big|\right)'\right) \ , \nn\\
 \tilde T_r &= \r^{-2}\del_r\tilde \rho 
 = -\frac{2 r^2 E B^2 S}{|A B|(B^2+r^2 S^2)}\left(\ln\Big|\frac{rS}{B}\Big|\right)' \ .
\end{align}
In particular,
the axion $\tilde\r$ is static if and only if
\begin{align}
 (A^2 - r^2 E^2)\left(\ln\Big|\frac{B}{r S}\Big|\right)' =  2 r^{-1} A^2 
 \ ,
\end{align}
or 
\begin{align}
 \left(\ln\Big|\frac{B}{r S}\Big|\right)' = \frac{2}{r} \frac{1}{1 - \frac{r^2E^2}{A^2}} \ .
 \label{axion-static}
\end{align}
Therefore for a static axion we obtain
\begin{align}
 \tilde T_r &= \r^{-2}\del_r\tilde \rho 
 = \frac{2 r^2 E B^2 S}{|A B|(B^2+r^2 S^2)}\frac{2}{r} \frac{1}{1 - \frac{r^2E^2}{A^2}} \ .
\end{align}
Taking into account the divergence constraint \eq{E-B-S-relation}, 
this reduces to
\begin{align}
 \tilde T_r &= \frac{4 r^4 \r^2 E^2 B^2 S}{c_e |A B|} \frac{1}{1 - \frac{r^2E^2}{A^2}} \ .
 \label{axion-static-Tr}
\end{align}

\paragraph{The equation of motion for the axion.}

It turns out that the equation of motion \eq{div-T-dEdE} for the axion  
holds identically for any spherically invariant configuration.
This can be seen explicitly in Cartesian coordinates, where 
the right-hand side of \eq{div-T-dEdE-form} is  obtained using \eq{epsdxdx-id}  as
\begin{align}
 T^{\dot\a}\wedge  T_{\dot\a} &= T^{i}\wedge  T^{j} \d_{ij} 
 = 4 \big(3 E S + r (S' E + S E') \big) dt d^3 x  \nn\\
 &= 4 S E\big(3 + r \frac{d}{d r}\ln|SE| \big) d^4 x
 = 4r S E\big( \ln|r^3SE| \big)' d^4 x
 ,
\end{align} 
while the left-hand side of \eq{div-T-dEdE-form} is obtained using \eq{axion-def} and \eq{axion-VF-general} as
\begin{align}
 d * (\rho^{-4}d\tilde{\rho})&=
 -\sqrt{|G|}\nabla^{(G)}_\mu (G^{\mu\nu}\r^{-4}\del_\nu\tilde\r) d^4x
 = -\del_\mu\big(\r^{-2}\sqrt{|G|}G^{\mu\nu}\tilde T_\nu\big) d^4x
 \nonumber \\
 &= -\frac 12 \del_\mu\big(\varepsilon^{\n\s\eta\mu}\tensor{T}{_\nu_\s_\eta}\big) d^4x
 = \del_i (4 S E x^i) d^4x
 = 4r S E\big( \ln|r^3SE| \big)' d^4 x
 \ . \nn 
\end{align} 
Therefore the equation of motion \eq{div-T-dEdE} for the axion  holds identically.

\section{Solving the geometric equations of motion}
\label{sec:solution-general}

\subsection{The $E\neq 0$ solutions}
\label{sec:E-neq0}

In this section, we derive the general solutions for $E\neq 0$,
using the divergence constraint \eqref{B-S-AE}.
For the most general ansatz \eq{ansatz}, 
the equation of motion \eq{torsion-frame-eom-3}
for $\dot\a=0$ and $\mu=r$
gives a second order differential equation, which
can be reduced to 
\begin{align} 
 \left( \r (\ln |A|)' \right)^2
 = \frac{c_0^2c_e}{2} E\left(\ln\Big|\frac{A}{r E}\Big|\right)' \ ,
 \label{rho-A-eq-0}
\end{align}
for an arbitrary real constant $c_0$.
Note that both sides of \eqref{rho-A-eq-0} are positive since
$(\ln|\frac{A}{r E}|)'$ is positive,
as seen from \eqref{E-B-S-relation} and \eqref{dilaton-relation-1},
and $c_eE$ is positive by definition \eqref{E-B-S-relation}.
Combining this with the divergence constraint \eq{dilaton-relation-1},
we obtain the relations\footnote{
The sign of $c_0$ is defined by \eqref{eom-1-AE}.
}
\begin{align}
 (\ln |A|)' 
 &= c_0 r B E \ , \nn\\
 (\ln|r E|)' 
 &=-\frac{2r^2 B^2 E \rho^2}{c_e}
 +c_0 r B E \ .
 \label{eom-1-AE}
\end{align}
Then the combination of the equation of motion for ${\dot\a} = 0$ and $\mu=t$
with \eq{rho-A-eq-0} implies 
$m=0$,
which is assumed as an approximation 
in the present static setup.
The equation of motion for ${\dot\a} = i$ and $\mu=t$
with the condition \eq{axion-static} for a static axion
leads to
\begin{align}
  (\ln|r B|)' &= - \frac 1r \frac{1}{r^2 \frac{E^2}{A^2} - 1} \ ,
  \label{axion-static-eq-1}
\end{align}
using the above relations \eq{rho-A-eq-0}  and \eq{eom-1-AE}.
Then the difference between \eqref{axion-static} and \eqref{axion-static-eq-1}
gives
\begin{align}
 (\ln|r^2S|)'=\frac{1}{r}\frac{1}{r^2\frac{E^2}{A^2}-1}
 \ ,
 \label{S-eq}
\end{align}
which can be written only by $A$ and $E$ using \eqref{B-S-AE}.
By eliminating the derivatives of $A, E$ and then $B$ in \eqref{S-eq}
using \eq{eom-1-AE} and \eq{axion-static-eq-1}, 
one obtains
\begin{align}
 (\ln\rho )'
 &= -\frac{1}{2} 
  \left(\frac{2 }{r^2\frac{E^2}{A^2}-1}\Big(\frac 1r - \frac{2r^2 B^2 E \rho^2}{c_e} \Big)
  - \frac{2r^2 B^2 E \rho^2}{c_e}
  + c_0 r B E\right) \nn\\
%
 &= -\frac{1}{r^2\frac{E^2}{A^2}-1}\Big(\frac 1r - \frac{2r^2 B^2 E \rho^2}{c_e} \Big)
 -\frac{1}{2} (\ln|r E|)'
 \ ,
 \label{eom-1-rho}
\end{align}
and hence  
\begin{align}
 (\ln|r E \rho^2| )' &=  -\frac{2}{r^2\frac{E^2}{A^2}-1}\Big(\frac 1r - \frac{2r^2 B^2 E \rho^2}{c_e} \Big)
 .
 \label{eom-1}
\end{align}
The equations of motion \eqref{torsion-frame-eom-3}
for ${\dot\a} = i$ and $\mu=r$ and 
those for any $\dot\alpha$ and $\mu=\vartheta$, $\varphi$
are automatically satisfied if the above equations are satisfied.

We have therefore obtained a coupled system of 4 nonlinear differential
equations for 4 functions $A(r),B(r),E(r)$ and $\r(r)$.
Solving such a system seems like a formidable task. Remarkably,
the general solution can be obtained rather explicitly.
To achieve this,
we combine the above equations to get
\begin{align}
 \left( \ln\Big|\frac{E \rho^2}{r B^2}\Big| \right)'
 &=  \frac{4}{r^2\frac{E^2}{A^2}-1} \frac{r^2 B^2 E \rho^2}{c_e}
 = -\left( \ln \left|\frac{A^2}{r^2 E^2} -1\right| \right)'
 ,
\label{eom-Erho2}
\end{align} 
using \eq{dilaton-relation-1} 
in the second step.
This can be integrated as   
\begin{align}
\frac{E \rho^2}{r B^2}\left(\frac{A^2}{r^2 E^2}-1\right) &= c_1 
\label{AE-1}
\end{align}
for some constant $c_1$, so that 
\begin{align}
 \boxed{\
 \frac{A^2}{r^2 E^2} = 1 + c_1 \frac{r B^2}{E \rho^2}
 \ . 
 }
 \label{AE-equation}
\end{align}
Then, \eqref{AE-equation} and \eqref{axion-static-eq-1} gives
\begin{align}
    (\ln |B|)' &= \frac{E \rho^2}{c_1 r^2 B^2}  \ .
\label{b2-deriv}
\end{align}
Together with \eq{eom-Erho2} we obtain
\begin{align}
%
 \frac 1{E \rho^2}\left(\ln\Big|\frac{E \rho^2}{r}\Big|\right)' + \frac{4}{c_e c_1} r E \rho^2 
 &= \frac{2}{c_1}\frac{1}{r^2 B^2} - \frac{4}{c_e} r^2 B^2 \ .
  \label{Erho-B-relation}
\end{align} 
Rewriting $E\rho^2$ by $B^2$ using \eq{b2-deriv}
leads to a second-order ordinary differential equation (ODE) for $B$, which is solved by
\begin{align}
\boxed{\ 
  -|B|^{4 c_3} \left((c_3+1) c_e \ \frac 1{r^4} + c_1 B^4\right) = c_2 
 \ , 
   \ }
  \label{B-equation}
\end{align}
for arbitrary integration constants $c_2$ and $c_3$.
Hence $B = B(r)$ satisfies a simple algebraic relation,
which can in fact be solved explicitly for $r$ as a function of $B$:
\begin{align}
 r^4 = \frac {-(c_3+1) c_e}{c_1 B^4  + c_2  |B|^{-4 c_3}}
 \ . 
 \label{r-B-equation}
\end{align}
The other equation for the integration constants $c_2$ and $c_3$
can be obtained by differentiation of \eq{B-equation}.
The differentiation leads to
\begin{align}
 (\ln |B|)' &= \frac 1r \frac{c_e}{c_1 r^4 B^4 + c_3c_e}
 ,
 \label{dB-equation}
\end{align}
for $B\ne 0$.
Together with \eq{b2-deriv},
one can derive an algebraic expression for $E\rho^2$
in terms of $B$ and $r$:
\begin{align}
 E \rho^2 =  c_1 c_e\frac{r B^2}{c_1\, r^4 B^4 + c_3 c_e} 
 \ .
 \label{Erho2-eq}
\end{align} 
Then $\frac{A^2}{E^2}$ is obtained explicitly from \eq{AE-equation} as 
\begin{align}
 \frac{A^2}{r^2 E^2} = \frac{1}{c_e}\big((1 + c_3)c_e + c_1\, r^4 B^4 \big)
 \ .
\label{AE-equation-2}
\end{align}
Using \eq{B-equation}, this can be written more compactly as
\begin{align}
\boxed{\ 
 r^6 E^2 = -\frac{c_e}{c_2} A^{2} |B|^{4 c_3}
 \ .
 }
\label{E2-B2-A2-relation}
\end{align}
It remains to solve one more equation for $E$ (or $A$), but this can no longer 
be achieved in algebraic form. However, we can 
combine \eq{axion-static-eq-1} and \eq{eom-1-AE} for $B$ and $E$ with the above results to obtain
\begin{align}
 (\ln|r^2 BE|)' &= - \frac 1r \frac{1}{r^2 \frac{E^2}{A^2} - 1}
 - \frac{2r^2 B^2 E \rho^2}{c_e}+ c_0 r B E \ .
\end{align}
Together with \eq{AE-equation} and \eq{Erho2-eq},
one obtains
\begin{align}  
 (\ln|r BE|)' - c_0 r B E
 &= - \frac 1r\frac{2c_1 r^4 B^4 - c_e}{c_1 r^4 B^4 + c_3 c_e} \ .
 \label{BE-B-relation-0}
\end{align}
This still involves  $B, E$ and $r$. However,
combining this equation with \eq{dB-equation} in the form
\begin{align}
 B\frac {d\ln|r B|}{dB} 
 &= 1 + c_3 + \frac{c_1}{c_e} (rB)^4 
 ,
 \label{drB-equation}
\end{align}
one can rewrite it as 
\begin{align}
 \frac {1}{(rE)^2} \frac{d(r E)}{d(r B)}
 &= c_0
  - \frac{c_0 c_e + 2c_1 (r B)^3 \frac{1}{rE}}{ c_e(1 + c_3) + c_1 (rB)^4 } \ ,
\end{align}
which is an ODE relating $(rB)$ and $(rE)$.
We now introduce the effective radius 
\begin{align}
 \tilde r^2 := \frac{c_e}{r E}
 = G_{\vartheta\vartheta}=\frac{G_{\varphi\varphi}}{\sin^2\vartheta}
 \ 
 \label{tilde-r-def}
\end{align}
(cf.~\eq{eff-metric-rtilde}),
which is positive since $E/c_e>0$. 
Then the above equation takes the form of a non-linear ODE 
relating $\tilde r^2$ and $(rB)$:
\begin{align}
 \frac{d\tilde r^2}{d(rB)} 
 &= \frac{2c_1 (r B)^3\, \tilde r^2+c_0c_e^2}{c_e(1 + c_3) + c_1 (rB)^4 } \ 
 -c_0c_e \ .
\label{dr2-drB}
\end{align}
This can be integrated in terms of the hypergeometric function 
$_2F_1(\frac{1}{4},\frac{1}{2};\frac{5}{4};z)$ as follows
\begin{align}
 \boxed{\
 \tilde r^2 
 =-\Big(\frac{(c_3+1) c_e}{c_1}\Big)^{\frac{1}{4}}
 \frac{c_0c_e z}{2(c_3+1)}\left((2 c_3+1)
 \sqrt{1+z^4} \, _2F_1\Big(\frac{1}{4},\frac{1}{2};\frac{5}{4};-z^4\Big)-1\right)
 +c_4 \sqrt{1 +z^4}
 \ ,
 \ }
 \label{rtilde-rB-solution}
\end{align}
where $c_4$ is a new integration constant arising from the homogeneous term, 
and we define
\begin{align}
 rB = \Big(\frac{(c_3+1) c_e}{c_1}\Big)^{\frac{1}{4}} \, z
 \ ,
 \label{rB_and_z}
\end{align}
for better readability. One can then verify that all components 
of the equation of motion 
are satisfied including
the equations \eq{Box-rho-onshell} and \eq{div-T-dEdE} for the dilaton and the axion, respectively.
We have therefore obtained the general solution for the case of $E\neq 0$.

Let us briefly take a look at some constraints on the functions and the parameters.
The physical regime, which we will focus on in the following, is
\begin{align}
 c_3\ge 0 \ ,
 \qquad c_ec_1>0 \ ,
 \qquad c_ec_2<0 \ .
 \label{signature-cond} 
\end{align}
These arise as follows:
As noted before, 
$c_e$ has the same sign as $E$ by definition, \eqref{E-B-S-relation}.
Then, $c_1$ and $c_e$ need to have the same sign 
so that the relation \eqref{AE-equation-2} at large $r$
is consistent for real functions $A$, $B$ and $E$.
Therefore, both signs of $c_e$ and $c_1$ match that of $E$.
This is consistent with a metric with the physical signature $G_{tt} < 0$ as appropriate for the 
large radius regime (see \eqref{eff-metric-Eneq0}).
$c_3>0$ follows from
\begin{align}
 S^2 = \frac {c_e}{r^5 E \r^2} - \frac {B^2}{r^2}
 = \frac{c_3 c_e}{c_1 r^6 B^2}
 \ ,
 \label{r2S2-eq}
\end{align}
which is obtained using \eqref{E-B-S-relation}
and \eqref{Erho2-eq}.
Then, the condition \eqref{signature-cond} properly implies $z^4 > 0$. 
In addition, 
\eqref{E2-B2-A2-relation} imposes that $c_e/c_2<0$.
Another constraint is that
$\ln|B|$ and $\ln|r B|$ are monotonically increasing functions.
The monotonicity of $\ln|B|$ is seen in \eqref{b2-deriv}.
Then from \eqref{drB-equation}, it turns out that 
$\ln|r B|$ is also monotonically increasing in $r$.
This implies $B$ is positive\footnote{
$B$ can change its sign only at $r=0$.
In the meantime,
$A$ can continuously change its sign at $r=0$ 
only if the sign of $(\ln |A|)'$ is positive as $r$ goes to positive infinity.
Nevertheless, we should discard the negative part of $A$ and/or $B$ 
since we focus on the physical radius, $r>0$.
} for any $r$ $(> 0)$.
Hence the physically meaningful region of $z$ is restricted to $z>0$.
Then, in the same manner, $A$ is positive for any positive $r$
because the sign of $(\ln |A|)'$ is fixed to the one same as $c_ec_0$ since $B>0$.

The solution clearly has an intricate analytic structure,
which should be explored in more detail elsewhere.

\paragraph{Asymptotic regime.}

For $r\to\infty$,  the frame should approach the background frame
\eq{background-frame-loc}.
This means that 
\begin{align}
 B(r), A(r), \r(r)  &\to \text{const} \neq 0, \qquad r\to\infty \ , \nn\\
 E(r), S(r),   &\to  0, \qquad \qquad \quad \ \ r\to\infty \ .
\end{align}
We then consider the asymptotic expansion for $r\to\infty$.
For $B(r)$, this is obtained from \eq{B-equation}
\begin{align}
\boxed{ \
 B = b_0 + \frac{b_4}{r^4} + \frac{b_8}{r^8}
 +O(r^{-12}),
 \qquad  b_0^{4 (c_3+1)} = - \frac{c_2}{c_1}
 \ ,
 \ }
 \label{b-expansion-explicit}
\end{align}
and the remaining coefficients $b_4, b_8, \cdots$ can be determined in terms of $c_e$, $c_1$ and $c_2$ if desired.
To proceed, we note that
the asymptotic behavior of the hypergeometric function is given by
\begin{align}
 z\, {}_2F_1\big(\frac{1}{4},\frac{1}{2};\frac{5}{4};-z^4\big)
 = \frac{\Gamma \left(\frac{1}{4}\right) \Gamma \left(\frac{5}{4}\right)}{\sqrt{\pi }}
 -\frac{1}{z} + O(z^{-5}),
 \qquad z \  \to \infty \ .
 \label{F21-asymptotic}
\end{align}
Then \eq{rtilde-rB-solution} simplifies as 
\begin{align}
\tilde r^2 
 &=-\Big(\frac{(c_3+1) c_e}{c_1}\Big)^{\frac{1}{4}}
 \frac{c_0c_e z}{2(c_3+1)}\Big((2 c_3+1) 
  \, \Big(\frac{\Gamma \left(\frac{1}{4}\right) \Gamma \left(\frac{5}{4}\right)}{\sqrt{\pi }} z - 1 \Big)-1\Big)
   + c_4 z^2 + O(z^{-2}) \nn\\
   &= z^2\big(\alpha_0 + \alpha_1 z^{-1} + O(z^{-4})\big)
   \label{rtilde-z}
\end{align}
for large $z$, where
\begin{align}
 z = \Big(\frac{(c_3+1) c_e}{c_1}\Big)^{-\frac{1}{4}} \, rB 
 = \Big(\frac{c_1}{(c_3+1) c_e}\Big)^{\frac{1}{4}} b_0 \, r \ + O(r^{-3})
 \ ,
\end{align}
using \eq{b-expansion-explicit}, and
\begin{align}
 \alpha_0
 &= -\frac{c_0c_e}{2}(2 c_3+1)
 \Big(\frac{c_e}{c_1(c_3+1)^{3}}\Big)^{\frac{1}{4}} \frac{ \Gamma \big(\frac{1}{4}\big) 
  \Gamma \big(\frac{5}{4}\big)} {\sqrt{\pi }}
 + c_4
 \ ,
 \nonumber \\
 \alpha_1
 &= c_0c_e \Big(\frac{(c_3+1)c_e}{c_1}\Big)^{\frac{1}{4}}
 \ .
 \label{a-b-def}
\end{align}
Together with \eq{tilde-r-def}, this gives 
\begin{align}
 \alpha_0 + \alpha_1 z^{-1} + O(z^{-4})
 = \frac{\tilde r^2}{z^2}
 = c_e \sqrt{\frac{(c_3+1)c_e}{c_1}} \frac{1}{(r E)(r B)^2}
 \ ,
 \label{z-rtilde}
\end{align}
and hence 
\begin{align}
  (r E)(r B)^2 = c_e \sqrt{\frac{(c_3+1)c_e}{c_1}} \frac{1}{\alpha_0+ \alpha_1 z^{-1}} + O(z^{-4}) \ .
  \label{EB2-expansion}
\end{align}
Then combining with
\eq{AE-equation-2} in the form
\begin{align}
 A^2 
 = \frac {c_1}{c_e}(rB)^4 (rE)^2 + O(r^{-4})
 \ ,
 \label{A-expand-1}
\end{align}
and assuming $A(r)$ approaches a positive constant $a_0$, at large $r$,
we obtain
\begin{align} 
 \boxed{\
 A
 =a_0 \left( 1
 +\frac{c_e}{|c_e|}\sqrt{\frac{c_e}{c_1}} \frac{a_0c_0}{b_0}
 \frac{1}{r}
 \right) ^{-1}
 + O(r^{-4})
 \ ,
 \ }
\label{A-expand}
\end{align}
and a relation of parameters through $\alpha_0$,
\begin{align}
 \frac{|c_e| \sqrt{c_3+1}}{a_0}
 =\alpha_0
= -\frac{c_0c_e}{2}(2 c_3+1)
 \Big(\frac{c_e}{c_1(c_3+1)^{3}}\Big)^{\frac{1}{4}} \frac{ \Gamma \big(\frac{1}{4}\big) 
  \Gamma \big(\frac{5}{4}\big)} {\sqrt{\pi }}
 + c_4
 \ ,
 \label{A-expand-a0}
\end{align}
assuming $\alpha_0>0$ since $\tilde r^2>0$.
Therefore, using \eqref{AE-equation-2} again, we obtain
\begin{align}
 \boxed{\
  r^3 E 
  = \frac{c_e}{\sqrt{c_ec_1}}\frac{a_0}{b_0^2}
  \left( 1
  +\frac{c_e}{|c_e|}\sqrt{\frac{c_e}{c_1}} \frac{a_0c_0}{b_0}
  \frac{1}{r}
  \right)^{-1}
  + O(r^{-4})
  \ ,
  \ }
  \label{E-asymptotics}
\end{align}
since $a_0>0$ and $E$ has the same sign as $c_e$. 
Next, $\rho^2$ is obtained from \eq{Erho2-eq} as
\begin{align}
 \boxed{\
 \rho^2  
 = \frac{\sqrt{c_ec_1}}{a_0}
 \left( 1
 +\frac{c_e}{|c_e|}\sqrt{\frac{c_e}{c_1}} \frac{a_0c_0}{b_0}
 \frac{1}{r}
 \right)
 + O(r^{-4})
 \ .
 \ }
 \label{rho-asymptotics}
\end{align} 
Finally, 
using \eqref{r2S2-eq},
we obtain 
\begin{align}
\boxed{\
 r^3 S = \pm \frac{1}{b_0 }\Big(\frac{c_3c_e}{c_1}\Big)^{\frac{1}{2}} \ + O(r^{-4})  \ .
 \ }
 \label{S-asymptotics}
\end{align}
We have thus determined explicitly the three leading terms of
asymptotic expansion of all the functions $A,B,E,S,\rho$ at $r\to\infty$.
Remarkably, the leading behavior is the same as in  the simple solution
\eq{special-solution-HS}, even though $S\sim O(r^{-3}) \neq 0$.
To meet the boundary conditions given by the background frame \eq{background-frame-loc},
 $a_0$ and $b_0$ are determined, with $b_0 = a_0$.
Furthermore,  $\rho^2$ must reduce to the background \eq{dilaton-BG}.
This provides three equations for the constants
$c_e,c_0,c_1,c_2,c_3,c_4$, since 
the other fields $E$ and $S$ vanish at $r\to\infty$ as required.
We therefore have three free parameters,
which specify the localized solution.
These presumably correspond to a mass parameter of the effective metric, and two extra scales for the dilaton and axion.

The constraints resulting from the boundary conditions can be made 
more explicit for the case
$A \to a_0 = 1, B\to b_0 = 1$ and $\rho\to\rho_0$ 
as $r\to \infty$,
where $\rho_0$ is set much larger than 1 as we are interested in the late time regime.
Then we obtain
\begin{align}
 c_2=-c_1,
 \qquad
 c_4
 =\rho_0^4\sqrt{\frac{c_3+1}{|c_1|}}\left(
 \frac{1}{\sqrt{|c_1|}}
 +\rho_0\frac{c_0(2 c_3+1)}{2c_1(c_3+1)^{\frac{5}{4}}}
 \frac{ \Gamma \big(\frac{1}{4}\big) 
 \Gamma \big(\frac{5}{4}\big)} {\sqrt{\pi }}
 \right) ,
 \qquad
 c_e=\frac{\rho_0^4}{c_1}
 .
\end{align}

\subsection{The $E = 0$ solutions}
\label{sec:E=0}

For $E=0$,
the divergence constraint 
reduces to only one equation \eqref{dilaton-constr-E=0}.
Then the condition $m=0$ can be derived from
the equations of motion \eqref{torsion-frame-eom-3}
again in a similar manner to the $E\neq 0$ case.
Then the equation of motion \eqref{torsion-frame-eom-3} 
for $\dot\a= i$ and $\mu=r$ gives
\begin{align}
r^4 B'+\left(\frac{r^6 S^2}{B}\right)' = 0 \ .
\label{eom-e0-2}
\end{align}
As in the case of $E\neq 0$, we also impose that the axion is static, $\tilde T_t = 0$,
which is equivalent to \eq{axion-static} with $E=0$, i.e.
\begin{align}
 \left(\frac{r^3 S}{B}\right)' = 0 \ .
\end{align}
This in turn implies via \eq{axion-VF-general}
that the axion is trivial, 
\begin{align}
 \tilde T_\mu = 0 = \del_\mu\tilde\rho 
 \ ,
\end{align}
and does not contribute to the energy-momentum tensor.
Together with \eq{eom-e0-2}, it follows that 
\begin{align}
 B &= b_0 \ , \qquad S = \frac{s_0}{r^3} \ ,
 \label{B-S-explicit-E0}
\end{align}
with constants $b_0$ and $s_0$,
where we assume $b_0>0$ since $B$ should reduce to the background \eq{background-frame-loc} at $r\to\infty$.
Then the divergence constraint 
is simplified as 
\begin{align}
 (\ln |A \rho^2|)' = \frac{2 s_0^2}{r^5 b_0^2+ r s_0^2}  \ .
 \label{Arho2-E0}
\end{align}
This is solved by
\begin{align}
 \r^2 = \frac 1{A}\frac{\tilde c_\rho}{\sqrt{b_0^2 + s_0^2 r^{-4}}} \ ,
 \label{rho-E0}
\end{align}
for some parameter $\tilde c_\rho$.
The equation of motion for $\dot\a=0$ and $\mu=t$
with \eq{rho-E0}
reduces to
\begin{align}
 0
 &= -\frac{2 b_0^2 r^3 }{b_0^2 r^4+ s_0^2}(\ln |A|)'
 - \frac{A''}{A}
 +2 \left(\frac{A'}{A}\right)^2
 \ .
 \label{Aprime-equ-E0}
\end{align}
This implies either $A'=0$, which will be recovered in \eq{Aprime-0},
or otherwise
\begin{align}
 0
 &= -\frac 12\left( \ln(b_0^2 r^4+ s_0^2) \right)'
 - (\ln|A'|)' +  2 (\ln |A|)'
 \ .
\end{align}
Hence
\begin{align}
 \Big(\frac{1}{A}\Big)'  &= \frac {c_{a1}}{\sqrt{b_0^2 r^4+ s_0^2}}
 \ ,
 \label{Aprime-elliptic}
\end{align}
where $c_{a1}$ is an arbitrary parameter.
This leads again to two branches:
$s_0=0$ and $s_0\neq 0$.

\paragraph{Special case $s_0=0$.}

In this case, \eq{Aprime-elliptic} reduces to 
$ \big(\frac{1}{A}\big)' = \frac {c_{a1}}{b_0 r^2}$, which leads to the solution
\begin{align}
 \frac{1}{A} &= -\frac {c_{a1}}{b_0 r} + c_{a2} \ , \qquad
 B = b_0 \ , \qquad  E = 0 = S \ ,
 \label{special-solution-HS-2}
\end{align}
for a new parameter $c_{a2}$,
and the dilaton constraint \eq{Arho2-E0} reduces to
\begin{align}
 A \r^2 =\frac{\tilde c_\rho}{b_0} 
 \ .
 \label{special-solution-HS-rho-2}
\end{align}
We have recovered precisely the solution 
\eq{special-solution-HS} found in \cite{Fredenhagen:2021bnw}.
Three of the four parameters $\tilde c_\r, b_0, c_{a1},c_{a2}$
are again determined by the boundary condition at $r\to\infty$,
which leaves one physical parameter.
This corresponds to the total ``mass'' of the 
solution, as discussed in some more detail in section \ref{sec:metric}.

\paragraph{Generic case $s_0 \neq 0$.}
 
In this case we can integrate \eq{Aprime-elliptic} as follows
\begin{align} 
 \frac{1}{A} &= \int\frac {c_{a1}}{\sqrt{b_0^2 r^4+ s_0^2}} dr 
  = c_{a2} + \frac{c_{a1}}{|s_0|} r \, _2F_1\left(\frac{1}{4},\frac{1}{2};\frac{5}{4};-\frac{b_0^2 r^4}{s_0^2}\right) \nn\\
  &= c_{a2} + 
  c_{a1}\frac{\Gamma \left(\frac{1}{4}\right) \Gamma \left(\frac{5}{4}\right)}{\sqrt{\pi b_0 |s_0|}}
  -\frac{c_{a1}}{b_0} \frac 1r + O\big( r^{-4}\big) 
 \ ,
 \label{Aprime-elliptic-2}
\end{align}
using \eq{F21-asymptotic} in the last step.
Moreover, \eq{rho-E0} leads to
\begin{align}
 \r^2 = \frac 1{A}\frac{\tilde c_\rho}{b_0} 
 + O\big( r^{-4}\big) \ .
\end{align}
We observe again the same asymptotic behavior as \eq{special-solution-HS} 
even though $S=s_0 r^{-3} \neq 0$.
Nevertheless the axion is  trivial and does not contribute to 
the energy-momentum tensor, as pointed out above. 
Three of the five parameters $\tilde c_\r, b_0, s_0, c_{a1},c_{a2}$
are again determined by the boundary condition at $r\to\infty$,
which leaves two physical parameters. These should correspond to the total 
``mass'' and one further scale $s_0$, which characterizes $S\neq 0$.
The meaning of the latter can be understood  from the special case 
 $c_{a1} = 0$, where the solution reduces to
\begin{align}
 \r^2 = \frac 1{a_0}\frac{\tilde c_\rho}{\sqrt{b_0^2 + s_0^2 r^{-4}}}\ , \qquad 
 B = b_0 \ , \quad A=a_0\ , \qquad S = \frac{s_0}{r^3} \ .
 \label{Aprime-0}
\end{align}
Three of the four parameters $a_0,b_0, \tilde c_\r, s_0$
are again determined by the boundary condition at $r\to\infty$,
which leaves only one physical parameter $s_0$.
Then the asymptotic mass parameter in the corresponding metric \eq{A-E0-metric} vanishes,
and the meaning of this solution remains to be understood.

For $A \to a_0 = 1, B=b_0 = 1$ and $\rho\to\rho_0$ as $r\to \infty$,
the constraints resulting from the boundary conditions are 
\begin{align}
 c_{a2}=1-c_{a1}\frac{\Gamma \left(\frac{1}{4}\right) \Gamma \left(\frac{5}{4}\right)}{\sqrt{\pi |s_0|}}
 ,
 \qquad
 \tilde c_\rho=\rho_0^2
 \ .
\end{align}

\paragraph{Remarks on cosmological solutions.}

Since we have found the most general rotationally invariant static solution,
it is natural to ask about possible cosmological solutions.
In principle there should indeed be more general cosmological solutions,
with an asymptotic behavior which is different from the asymptotically 
flat case under consideration here.
It should be possible to study them systematically using a suitably adapted 
framework; however then the restriction to the static case must be relaxed.
This is the reason why we have not obtained such cosmological solutions,
and we leave that for future work.

\section{The effective metric}
\label{sec:metric}

For the spherically symmetric ansatz under consideration, the
effective metric  takes the form    \eq{eff-etric-general},
or equivalently
\begin{align}
  ds^2_G &=  \r^2\big(-(A^2-r^2 E^2) dt^2 + B^2 dr^2 
   + 2 r B E dt dr
  + (B^2+r^2 S^2) r^2d\Omega^2 \big)
  \ ,
  \label{eff-etric-general-2}
\end{align}
where $d\Omega^2 = d\vartheta^2 + \sin^2\vartheta d\varphi^2$.
In this section, we will elaborate this metric more explicitly for the solutions
found above.

\subsection{The $E\neq 0$ branch}
\label{sec:E-neq0-metric}

Using the divergence constraint \eq{E-B-S-relation} and the on-shell equation \eq{AE-equation},
the metric can be written as 
\begin{align}
  ds^2_G &=  -  c_1 r^3 E B^2 dt^2 + \r^2 B^2 dr^2 + 2 \r^2 r B E dt dr
  + \frac{c_e}{r^3 E}\, r^2 d\Omega^2 \ .
  \label{eff-metric-Eneq0}
\end{align}
The off-diagonal term can be eliminated by a suitable redefinition 
$t = \tilde t + \psi(r)$ 
with
\begin{align}
 \psi' = \frac{\r^2 r B E}{c_1 r^3 E B^2} = \frac{\r^2}{c_1 r^2 B}
 .
 \label{psi-redefinition}
\end{align}
Then we obtain the metric in a diagonal form 
\begin{align}
  ds^2_G
  &= - c_1 r^3 E B^2 d\tilde t^2 
  + \r^2 \left(B^2 + \frac{1}{c_1}\frac{\r^2 E }{r}\right) dr^2
  + \frac{c_e}{r^3 E}\, r^2 d\Omega^2
\end{align}
(we will drop the tilde on $\tilde t$ henceforth).
The standard signature, $G_{tt}<0$, $G_{rr}>0$, $G_{\theta\theta}>0$, is
guaranteed as long as \eqref{signature-cond} holds.
Furthermore,
defining the effective radial variable $\tilde r$ as
\begin{align}
 \tilde r^2 
 = \frac{c_e}{r E}
 \ ,
 \label{r-tilde-def-Enon0}
\end{align}
which is positive for any $c_e$, 
we can bring the  metric to the following normal form
\begin{align}
 ds^2_G 
 &=  -c_1 c_e\frac{(rB)^2}{\tilde r^2} d t^2 
 + \r^2 \left(B^2 + \frac{1}{c_1}\frac{\r^2 E }{r}\right) 
 \left(\frac{dr}{d\tilde r}\right)^2 d\tilde r^2
 + \tilde r^2 d\Omega^2  \nn\\
 &= G_{tt} dt^2 + G_{\tilde r\tilde r}d\tilde r^2
 + \tilde r^2 d\Omega^2  \ .
 \label{eff-metric-rtilde}
\end{align}
Note that $-\tilde r^2 G_{tt}$ is always positive for any $r$
as long as $\tilde r^2>0$ is satisfied.
This means that 
there is no way to recover the Schwarzschild 
geometry from the $E\neq 0$ solutions.
There might be a 
radius where $G_{tt}=0$ due to  $B=0$, 
but the sign of $G_{tt}$ will never change. 


To make the radial metric $G_{\tilde r\tilde r}$ more explicit,
we need
\begin{align}
 \frac{\tilde r'}{\tilde r} = \frac{1}{\tilde r} \frac{d\tilde r}{d r} 
 &=  - \frac 1{2r}\Big( -\frac{2r B \rho^2}{c_e} + c_0 \Big)
 (r E)(r B)
 \ ,
\end{align}
which is obtained  using \eq{eom-1-AE}.
Therefore 
\begin{align}
 G_{\tilde r\tilde r}  
 &= \r^2 \left( B^2 + \frac{1}{c_1}\frac{\r^2 E }{r} \right)
 \left( \frac{dr}{d\tilde r} \right)^2
 = \frac{4\r^2}{(rB)^2} \frac{ (rB)^2\tilde r^2 + \frac{c_e}{c_1}\r^2}
 {\big( 2(rB) \rho^2 - c_0 c_e\big)^{2}} \ .
\end{align}
This can be made more explicit using \eq{Erho2-eq}
\begin{align}
 \rho^2 
 = \frac{c_1 r^2 B^2}{c_1 r^4 B^4 + c_3 c_e} \tilde r^2
 \ .
\label{rho-explicit}
\end{align}

Clearly $\rho^2 \to \text{const}$ for $r\to\infty$, as it should.
It seems that $G_{\tilde r\tilde r}$ is regular unless 
$rB\rho^2=c_0 c_e/2$ holds, at which $\tilde r'=0$, 
or $rB=0$ 
holds if $c_3=0$.


Upon inverting the relation \eq{rtilde-rB-solution} 
between $\tilde r$ and $z$, 
the metric is fully determined as a function of $\tilde r$.
We can make this more explicit in the asymptotic regime $\tilde r\to\infty$.

Some representative plots are shown in Fig.~\ref{fig:metric_rt2_Eneq0}
against the variable $z$ defined in \eqref{rB_and_z}.
In these graphs, we set the asymptotic behavior of the metric
as $A\to 1$, $B\to 1$, and $\rho\to \rho_0=3$.
The physically meaningful region is  $\tilde r^2\geq 0$ and $z>0$.
However, while $\tilde r^2$ 
in the center graphs and the bottom-left graph in Fig.~\ref{fig:metric_rt2_Eneq0}
monotonically increases in the region where $\tilde r^2>0$,
it is not monotonic in the other graphs. 
In particular, $\tilde r^2$ in the top-left and bottom-right graphs have a minimum.
These different behaviors of $\tilde r^2$ are depicted in Fig.~\ref{fig:rt2_Eneq0},
the condition of which can be read off from \eqref{dr2-drB} in principle.

In the case with a minimum of $\tilde r^2$,
since we assume $B\to 1$ at large $r$,
the physical region should be the one where $z$ is positive at large $\tilde r$.
The physical meaning of the other region, where $z$ becomes smaller,
is not clear yet. 
Since the radial parameter $\tilde r$ grows 
in both directions, the metric is reminiscent of a wormhole, which could be linking the two sheets of the cosmic background \cite{Sperling:2019xar}.
This is consistent with the observation that $\tilde r^2$ tends to be 
strictly positive, and $\tilde r\to 0$ only if $E\to\infty$. 

Another  interesting observation,
which can be understood from \eqref{eff-metric-rtilde},
is that
$G_{tt}$ diverges as $\tilde r$ approaches zero as seen in the top-right, center and bottom-left graphs;
meanwhile, $G_{\tilde r\tilde r}$ diverges 
as the first derivative of $\tilde r^2$ approaches zero as seen in the top and bottom-right graphs.

Note that $G_{tt}$ approaches $\rho_0^2$ at large $z$
while $G_{\tilde r\tilde r}\sim 1$ 
because of the variable transformation \eqref{psi-redefinition}. 
\begin{figure}[htbp]
 \centering
 \includegraphics[scale=0.7]{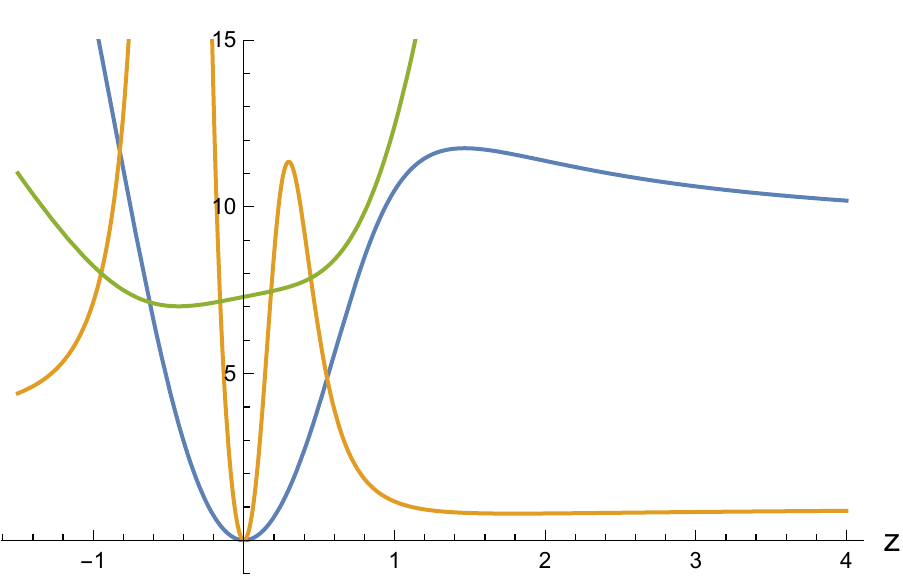}
 \includegraphics[scale=0.7]{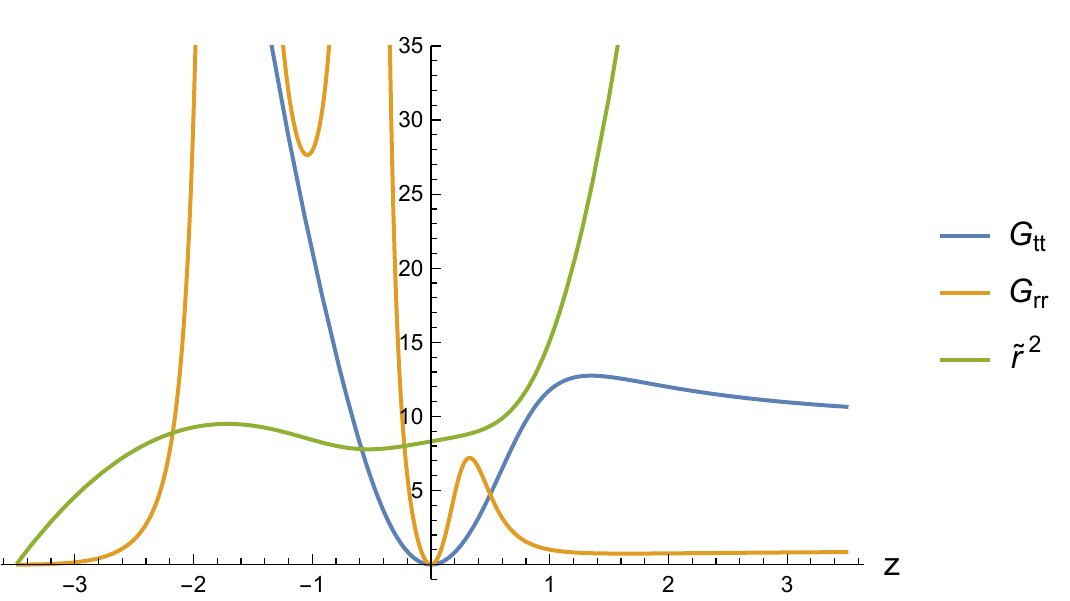}
 \includegraphics[scale=0.7]{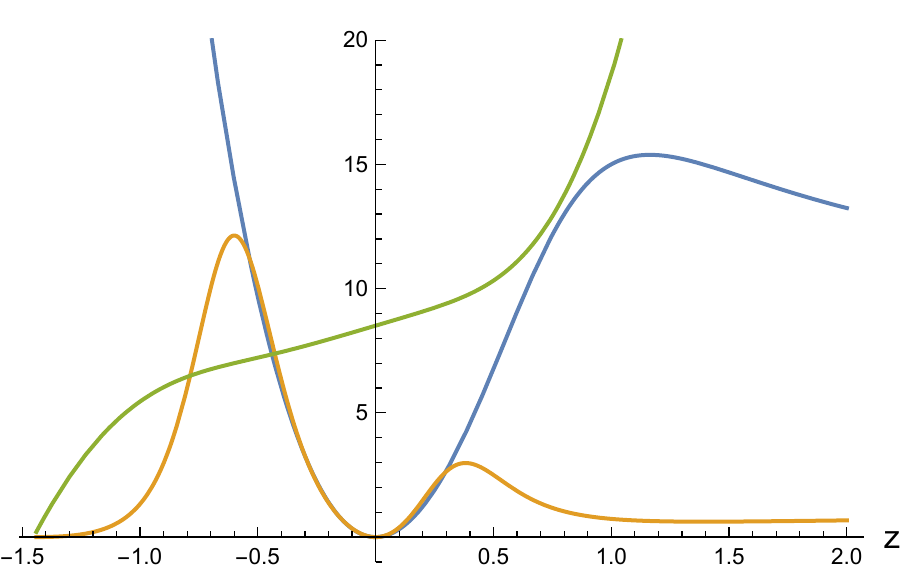}
 \includegraphics[scale=0.7]{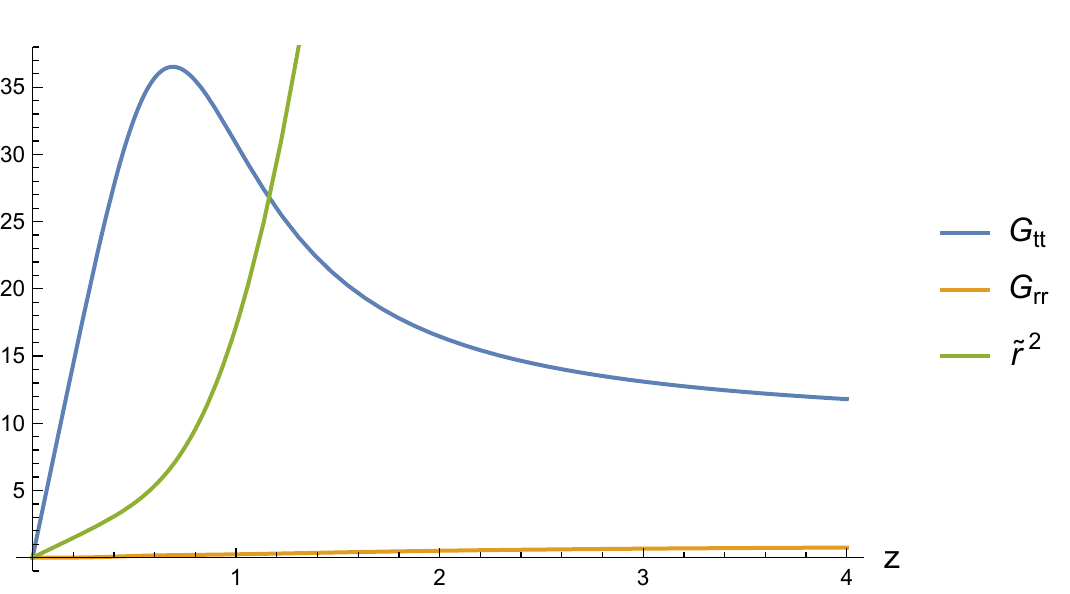}
 \includegraphics[scale=0.7]{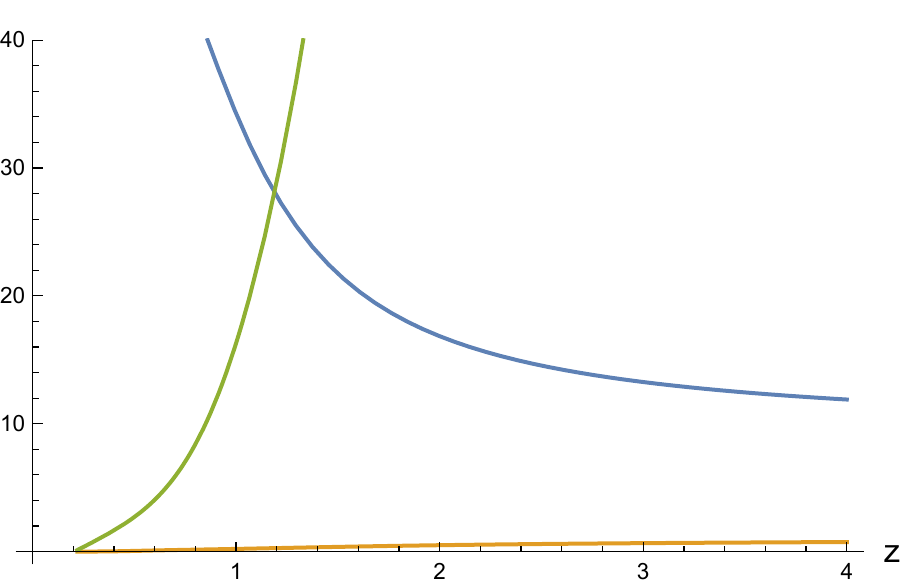}
 \includegraphics[scale=0.7]{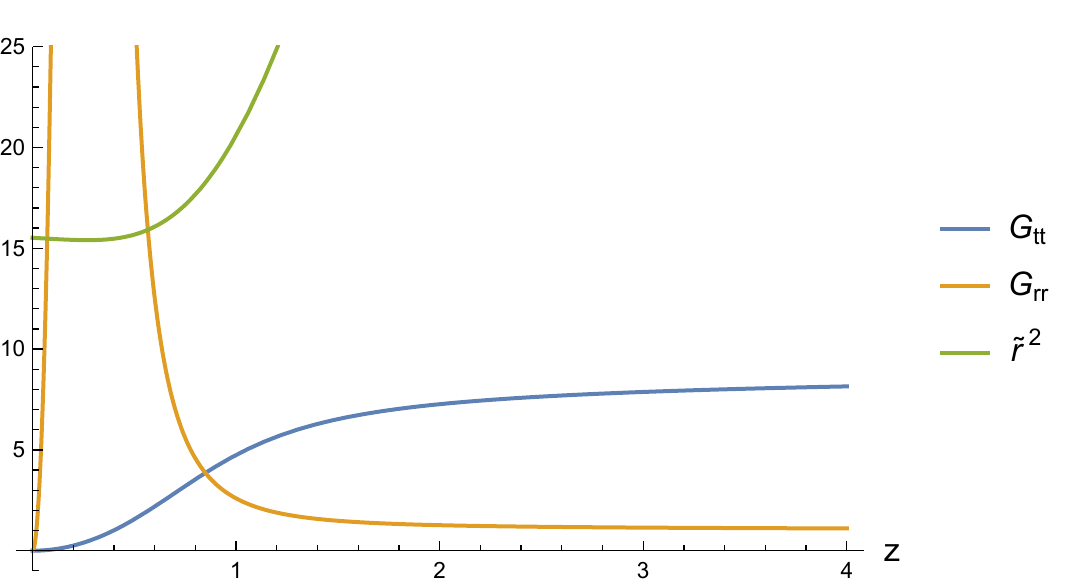}
 \caption{Graphs of the metric $G_{tt}$, $G_{\tilde r\tilde r}$ and $\tilde r^2$. We set $c_3=0.15$, $c_0=0.4$ and $\rho_0=3$. The top-left, top-right, center-left, center-right, bottom-left and bottom-right plots are of $c_1=-6$, $-4.4$, $-2.8$, $-1.47\cdots$, $-1.4$ and $8$, respectively. The center-right graph with $c_1=-1.47\cdots$ is a special case in which $c_4=0$. Note that some of the plots show behavior in $z<0$ though such a region is not physically meaningful.}
 \label{fig:metric_rt2_Eneq0}
\end{figure}

\begin{figure}[htbp]
 \centering
 \includegraphics[scale=0.8]{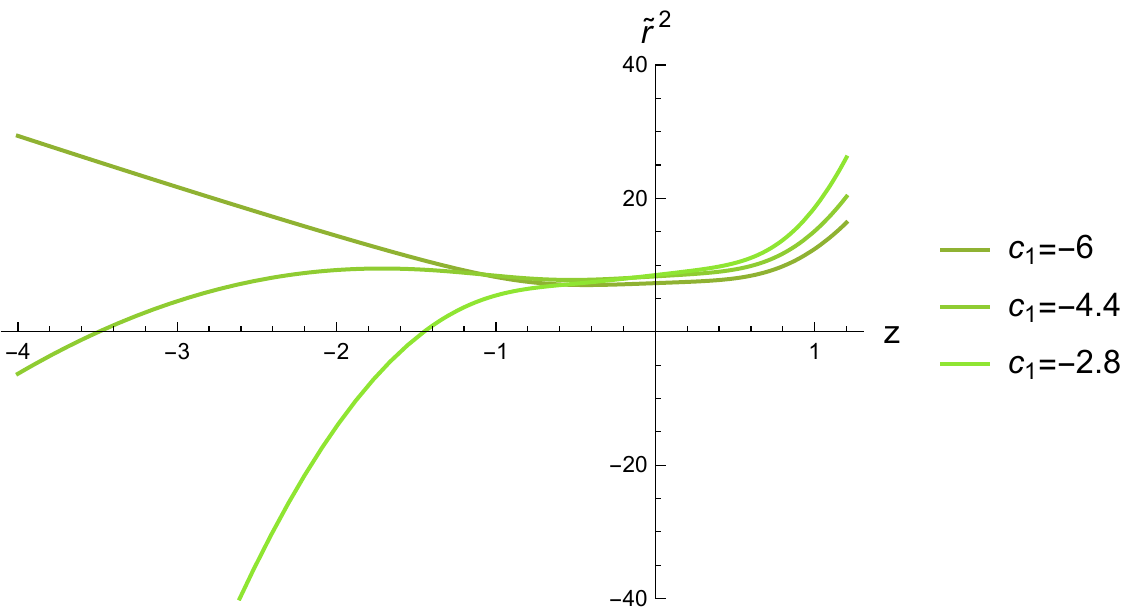}
 \caption{Three types of behaviors of $\tilde r^2$.}
 \label{fig:rt2_Eneq0}
\end{figure}

\paragraph{Asymptotic behavior.}

The long-distance asymptotics of the metric is obtained most easily by recalling that
using $E \sim r^{-3} \sim S$ \eq{E-asymptotics}, \eq{S-asymptotics}. Therefore
\begin{align}
  ds^2_G &=  \r^2\Big(-A^2 dt^2 + B^2 dr^2 
   + 2 r B E dt dr + B^2 r^2 d\Omega^2
  \Big)\ + O(r^{-4}) \ .
\end{align}
By redefining $t\to t+\psi'(r)$ and
using $B = b_0 + O(r^{-4})$ \eq{b-expansion-explicit}, 
we obtain
\begin{align}
  ds^2_G &=  \r^2\Big(-A^2 dt^2 + b_0^2 (dr^2 + r^2 d\Omega^2)\Big)\ + O(r^{-4}) \ .
\end{align}  
Moreover, 
\eqref{A-expand} and \eqref{rho-asymptotics} give
\begin{align}
 \rho^2 A
 &= \sqrt{c_ec_1} + O(r^{-4})
 .
\end{align}
Therefore the metric has the asymptotic form
\begin{align}
 ds^2_G
 &= \sqrt{c_ec_1}\left(
 -A dt^2 + \frac{b_0^2}{A}(dr^2 + r^2 d\Omega^2)
 \right) + O(r^{-4})
 ,
\end{align}
where
\begin{align}
 \frac 1A
 =
 \frac{1}{a_0}
 \Big(1 + \frac {2M}r \Big) + O(r^{-4})
 \label{1/A-metric}
\end{align}
due to \eq{A-expand}, with mass parameter
\begin{align}
 M
 =\frac{c_e}{2|c_e|}\sqrt{\frac{c_e}{c_1}} \frac{a_0 c_0}{b_0}
 \ .
\end{align}
This has the same structure as the simple solution \eq{special-solution-HS} found in \cite{Fredenhagen:2021bnw},
and reproduces the linearized Schwarzschild metric with mass $M$.
Note that we have to choose $c_ec_0>0$
in order to describe a positive mass.

We can compare the metric with the standard Eddington-Robertson-Schiff 
parameters
\begin{align}
 ds_{SS}^2 = -\Big(1-\frac{2M}{r} + 2\b \frac{M^2}{r^2} + ...\Big) dt^2 
 + \Big(1 + 2\g \frac{M}{r}  + ...\Big) (dr^2 + r^2 d\Omega^2)
 ,
\end{align}
which in GR take the values $\b=\g=1$,
while the present solution corresponds to $\g=1$ but $\b=2$.
This means that some of the solar system precision test are not satisfied.
However, this is not surprising since
we have not  taken into account the Einstein-Hilbert term,
which is induced in the quantum effective action at one loop \cite{Steinacker:2021yxt}. 
Since the E-H action has two extra derivatives compared with the bare matrix model 
action\footnote{The E-H action is quadratic in the torsion 
$T^{\dot\a \dot\b \mu} =- \{\Theta^{\dot\a\dot\b},x^\mu\} \sim \theta^{\mu\nu}\del_\nu\Theta^{\dot\a\dot\b}$, while the 
matrix model action is quadratic in $\Theta^{\dot\a\dot\b} = \{Z^{\dot\a},Z^{\dot\b}\}$, cf.~\cite{Steinacker:2021yxt}.},
it is plausible that the induced E-H action will dominate for short distances, while the present solution of the 
classical matrix model should dominate at (very) long distances. Then the above metric should perhaps be 
compared with the metric on galactic scales rather than solar system scales,
and the above deviation from Ricci flatness might be compatible with the observation of 
galactic rotation curves. This will be briefly discussed  in section \ref{sec:rotation-curves}.

\subsection{The $E=0$ branch}
 
For $E=0$, the  effective metric \eq{eff-etric-general-2} takes the form
\begin{align}
  ds^2_G
  &=  \r^2\left( -A^2 dt^2 + b_0^2 dr^2 
  + \left( b_0^2 + \frac{s_0^2}{r^4} \right) r^2 d\Omega^2 \right) \nn\\
  &= \tilde c_\rho\Big(- \frac{A}{\sqrt{b_0^2 + s_0^2 r^{-4}}}  dt^2 
  +  \frac{b_0^2}{A\sqrt{b_0^2 + s_0^2 r^{-4}}}\frac{d\tilde r^2}{(\tilde r')^2}
  + \tilde r^2 d\Omega^2\Big)
  \ ,
  \label{A-E0-metric}
\end{align}
using \eq{B-S-explicit-E0} and \eq{rho-E0}.
Here $A$ is given explicitly by \eq{Aprime-elliptic-2}, and
we introduced again the
effective radial variable $\tilde r$  via
\begin{align}
 \tilde r^2 
 = \frac{\sqrt{b_0^2 r^4 + s_0^2}}{A} \ ,
 \label{r-tilde-def-E0}
\end{align}
which allows to express $\tilde r'$ using \eq{Aprime-elliptic} as follows
\begin{align}
2\tilde r \tilde r' 
 &=  c_{a1}
 + \frac{2b_0^2 r^3 }{A\sqrt{b_0^2 r^4 + s_0^2}} \ .
\end{align}
Some representative plots are shown in Fig.~\ref{fig:metric_rt2_Eeq0}.
Again, we set the asymptotic behavior of the metric
as $A\to 1$, $B\to 1$, and $\rho\to \rho_0=3$.
$\tilde r^2$ is positive and monotonically increasing in the upper graphs
with $c_{a1}=1$, $0.14$,
but it has a minimum in the lower graph with $c_{a1}=-0.5$.
As is the case for $E\ne 0$, $G_{tt}$ diverges as $\tilde r$ approaches zero, 
and $G_{\tilde r\tilde r}$ diverges 
as the first derivative of $\tilde r^2$ approaches zero. 
$G_{tt}$ and $G_{\tilde r\tilde r}$ approach $1$ at large $r$.
\begin{figure}[htbp]
 \centering
 \includegraphics[scale=0.7]{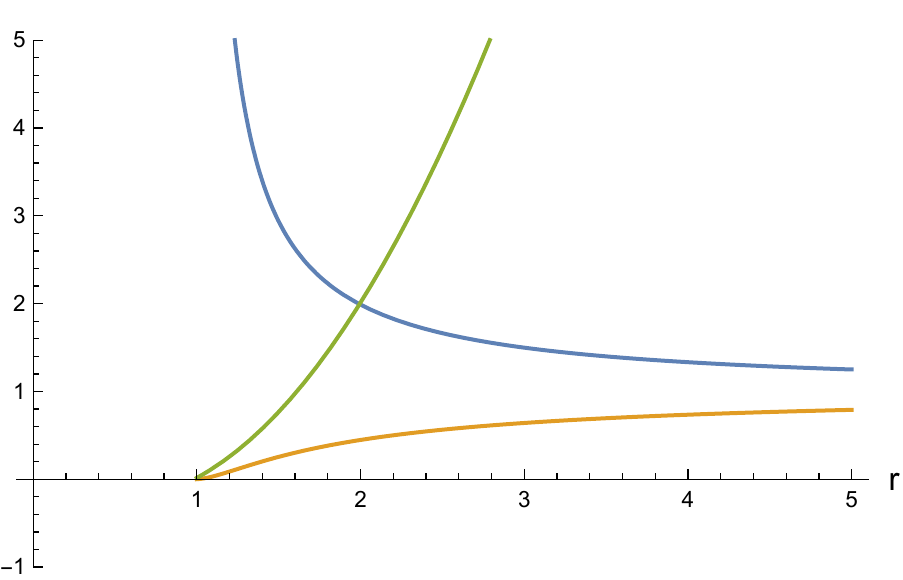}
 \includegraphics[scale=0.7]{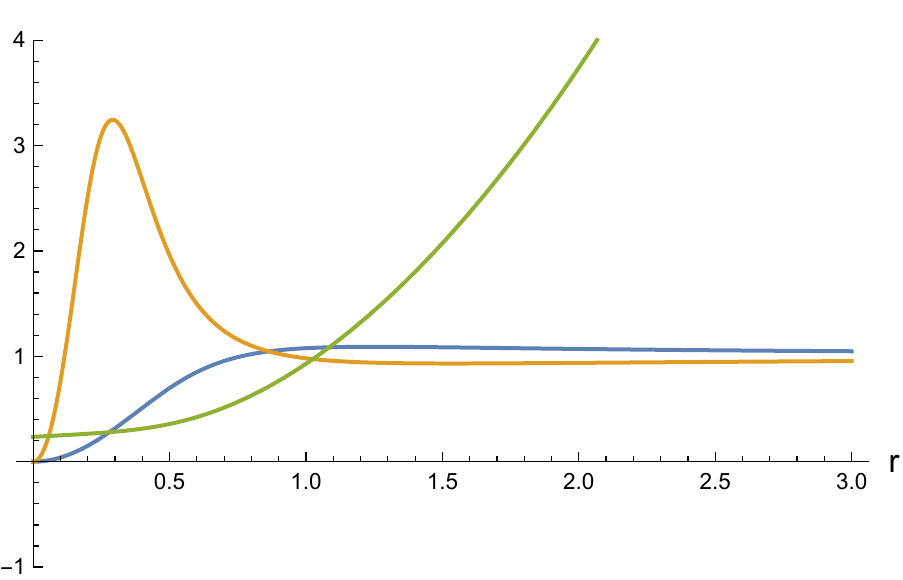}
 \includegraphics[scale=0.7]{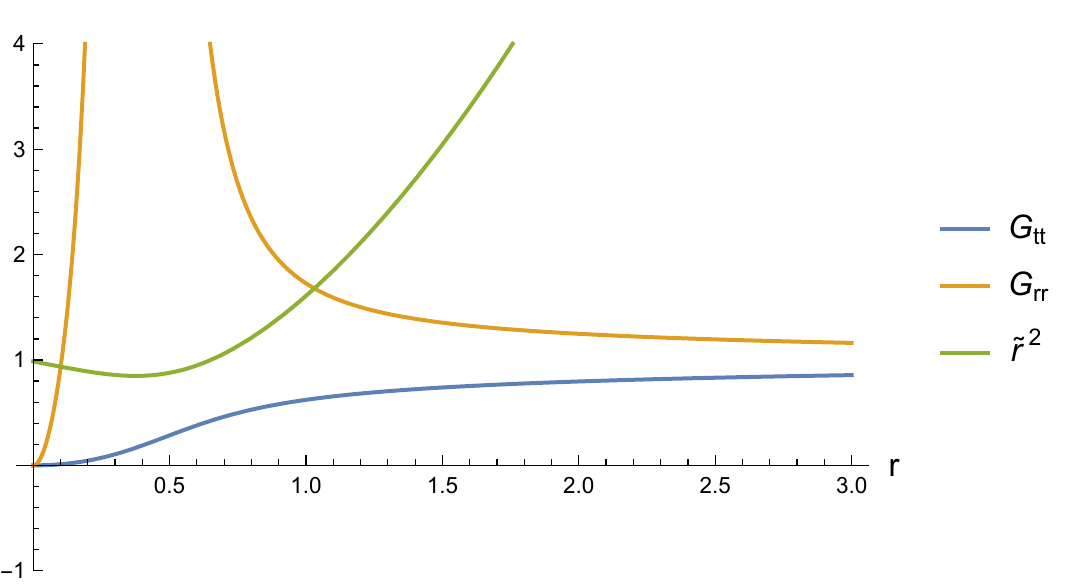}
 \caption{Graphs of the metric $G_{tt}/\tilde c_\rho$, $G_{\tilde r\tilde r}/\tilde c_\rho$ and $\tilde r^2$. We set $s_0=0.4$ and $\rho_0=3$. The top-left, top-right and bottom plots are of $c_{a1}=1$, $0.14$ and $-0.5$, respectively.}
 \label{fig:metric_rt2_Eeq0}
\end{figure}

\paragraph{Asymptotic behavior.}

The long-distance asymptotics of this metric is given by
\begin{align}
 ds^2_G
 &= \r^2A\Big(-A dt^2 + \frac{b_0^2}{A} (dr^2  + r^2 d\Omega^2) \Big)
 + O(r^{-4})  \ .
\end{align}
Using \eqref{special-solution-HS-2} and \eqref{Aprime-elliptic-2} 
we obtain
\begin{align}
 \frac{1}{A}
 &= \tilde c_{a2}\left(
 1+\frac{2M}{r}
 \right)
 + O(r^{-4})
 \ ,
 \nn\\
 \r^2 A &= \frac{\tilde c_\rho}{b_0} \ + O(r^{-4})
 \ ,
\end{align}
where
\begin{align}
 &\tilde c_{a2}=\left\{
 \begin{array}{ll}
  c_{a2}
  & \text{; } s_0=0 \\
  c_{a2}
  +c_{a1}\frac{\Gamma(\frac{1}{4}) \Gamma(\frac{5}{4})}
  {\sqrt{\pi b_0 |s_0|}}
  & \text{; } s_0\neq 0
 \end{array}
 \right.
 \nonumber \\
 &M=-\frac{c_{a1}}{2\tilde c_{a2}b_0}
 \ ,
\end{align}
where the terms of $O(r^{-4})$ is absent if $s_0=0$.
This has again the same structure as the simple solution \eqref{special-solution-HS} \cite{Fredenhagen:2021bnw} 
up to $O(r^{-4})$, 
which deviates from Ricci-flat at the nonlinear level.

\subsection{Rotation curves}
\label{sec:rotation-curves}

We expect that the IR regime of the classical computation is more trustworthy than 
the short-distance regime, where the quantum effects 
(such as an induced Einstein-Hilbert action \cite{Steinacker:2021yxt}) may be important.
It is therefore
interesting to consider the rotation velocities of circular orbits
in these metrics in the large $r$  regime. 
The question is if this may be similar to the 
rotation curves $v(r)$ observed in galaxies, whose well-known 
 deviation from GR is usually attributed to dark matter.

In the non-relativistic linearized
approximation, the effective (Newtonian) gravitational potential $V(r)$ is
given by
\begin{align}
 G_{tt} = -(1+V(r))
 \ .
\end{align}
The equation of motion for a stationary circular orbit is
\begin{align}
 m \frac{v^2(r)}{r} = m V'(r), \qquad v(r) = \sqrt{r V'(r)}
\end{align}
(assuming that $r$ is the effective distance).
For a central mass in Newtonian gravity, the potential 
$V(r) = -\frac{2M}{r}$ leads to 
\begin{align}
 v(r) \sim \sqrt{\frac{2M}{r}}
 \ .
\end{align}
Now consider the present solution. We have seen that in the long-distance regime $r\to\infty$, 
the metric has the universal form 
\begin{align} 
  ds^2_G &\sim -A dt^2 + \frac{1}{A}(dr^2 + r^2 d\Omega^2)\ + O(r^{-4})
\end{align}
up to an overall factor,
where
\begin{align} 
 \frac 1A  = 1 + \frac {2M}r \ + O(r^{-4})
 \ .
 \label{1/A-metric-rot}
\end{align}
Here we set $a_0=b_0=1$.
Then
\begin{align}
 V(r) \sim \frac{1}{1 + \frac {2M}r} - 1 = -\frac{2M}{r + 2M}
 .
 \label{grav-potential-rot}
\end{align}
Then we obtain
\begin{align}
 v(r) = \sqrt{\frac{2Mr}{(r + 2M)^2}} = 
 \sqrt{2M}\frac{\sqrt{r}}{r + 2M}
 .
 \label{rot-curve-formula}
\end{align}
For $r\gg M$, this reproduces the Newtonian $v(r) \sim \sqrt{r^{-1}}$, as it must.
However in the regime $r\sim 2M$, the rotation curve is indeed approximately flat, 
cf.~Fig.~\ref{fig:rot-curves}.
Of course for aspherical rotating objects, the resulting rotation curve would 
look somewhat different, and presumably stretched in the rotation plane.

\begin{figure}
 \centering
 \includegraphics[width=0.4\textwidth]{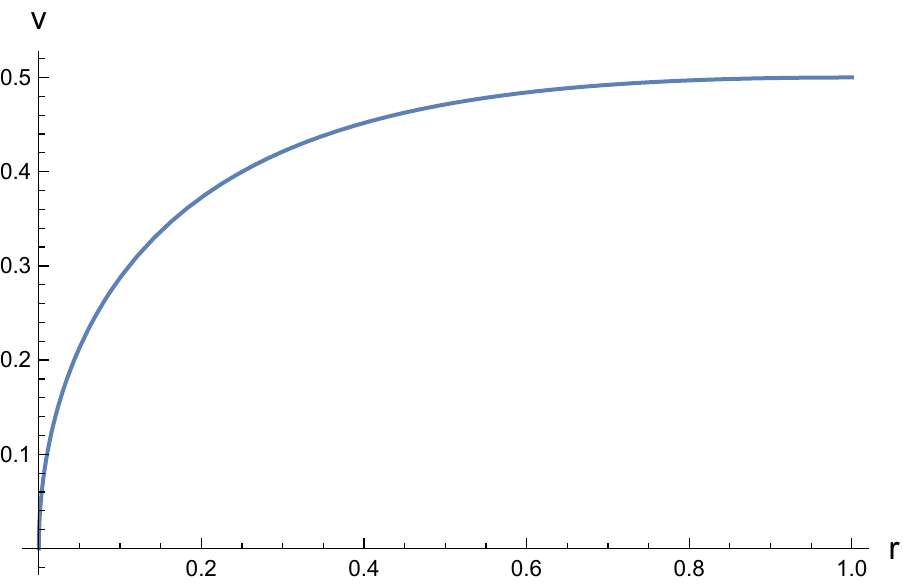}
 \includegraphics[width=0.4\textwidth]{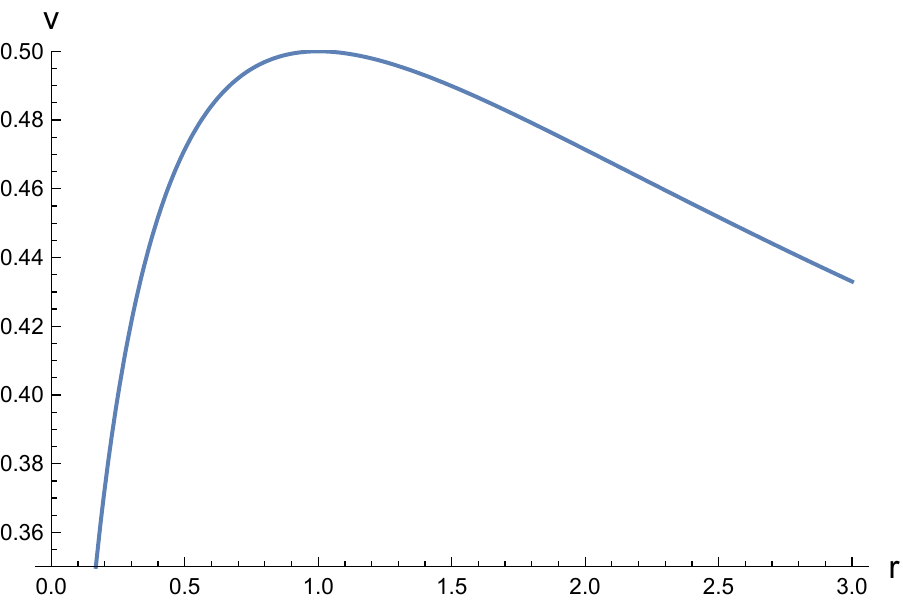}
 \caption{A rotation curve $v(r)$ arising from \eq{rot-curve-formula},
 with different truncation in $r$. $M$ is set to $1/2$.}
 \label{fig:rot-curves}
\end{figure}
Zooming into the appropriate regime, this may indeed
looks like a flat rotation curve. At shorter distances, more pronounced and 
localized feature may arise e.g.~from a non-vanishing axion.

It is amusing to compute the hypothetical mass distribution $M_N(r)$
of ``dark matter'' which would 
result in the same rotation curve in Newtonian gravity. 
This is given by 
\begin{align}
 M_N(r) =  M\frac{r^2}{(r + 2M)^2}
 \ ,
\end{align}
where
\begin{align}
 v(r) = \sqrt{\frac{2M_N(r)}{r}} 
 \ .
\end{align}
This seems not entirely unreasonable.

Note that there is no hidden Newton constant in the potential 
\eq{grav-potential-rot}, and 
there is a priori no reason to expect that $M=G_N m$ where $m$ is the 
physical mass in the center; rather, this long-range ``tail'' of the gravitational 
potential may be related to $m$ in some other, indirect way. 
The maximum of the velocity function \eq{rot-curve-formula} is at
\begin{align}
 r_M = 2M \ ,
\end{align}
where $V(r_M) = -\frac 12$. Hence $r_M$ is in a regime where the metric is significantly different 
from its asymptotic value $V(\infty) = 0$.
This corresponds to the strong gravity regime with associated redshift 
$\frac 12$,
which is completely unrealistic for galaxies. Therefore this 
simple picture does not work. A more complete and perhaps realistic analysis would require incorporating 
the induced Einstein-Hilbert action into the present model. Then 
it is conceivable that a similar effect arises from a cross-over between the 
Einstein-Hilbert regime at short scales and the matrix regime 
at longer distances. However, this would require a more sophisticated analysis. The main point here is that the matrix model framework  admits vacuum geometries which deviate from Ricci-flatness 
at large scales.

\subsection{The effective energy-momentum tensor}

We can now compute explicitly the various contributions to the effective energy-momentum tensor in the effective Einstein equations arising from 
the semi-classical matrix model \cite{Steinacker:2020xph}:
\begin{align}
 \cR_{\mu\nu}  - \frac 12 G_{\mu\nu} \cR  = {\bf T}_{\mu\nu}
 \ .
 \label{Einstein-eq-vac}
\end{align}
Here $\cR_{\mu\nu}$ and $\cR$ are the Ricci tensor and the scalar curvature
of the effective metric $G_{\mu\nu}$, respectively,
and the energy-momentum tensor is
\begin{align}
  {\bf T}_{\mu\nu}  &=  {\bf T}_{\mu\nu}[E^{\dot\a}] 
   + {\bf T}_{\mu\nu}[\r] 
    +  {\bf T}_{\mu\nu}[\tilde\r] 
       - \r^{-2} m^2 G_{\mu\nu}
 \ .
\end{align}
The contributions from the dilaton $\rho$, the axion $\tilde\rho$
and the frame $E^{\dot\alpha}$ are
\begin{align}
 {\bf T}_{\mu\nu}[\r] &= 2\r^{-2}\Big(\del_\mu\r \del_\nu\r - \frac 12 G_{\mu\nu}G^{\sigma\sigma'}\del_\s\r \del_{\s'}\r \Big)
 \ ,
 \label{e-m-tensor-rho}
 \\
 {\bf T}_{\mu\nu}[\tilde\r] 
 &= \frac 12 \r^{-4}\Big(\del_\mu\tilde\r \del_\nu\tilde\r 
 - \frac 12 G_{\mu\nu} G^{\s\s'}\del_\s\tilde\r \del_{\s'}\tilde\r \Big)
 \ ,
 \\
 {\bf T}_{\mu\nu}[E^{\dot\a}] 
 &= \rho^{2}\Big(
 -\tensor{T}{_\mu_\r^\s}\tensor{T}{_\nu_{\rho'}_\s}G^{\rho\rho'}
 + \frac 14 G_{\mu\nu} \tensor{T}{_\rho_\sigma^\kappa}\tensor{T}{_{\rho'}_{\sigma'}_\kappa} G^{\rho\rho'}G^{\sigma\sigma'}
 \Big)
 \ .
 \label{e-m-tensor-frame}
\end{align}

\paragraph{Axion.}

Consider first the axion, which is non-vanishing only 
for $E\neq0$.
Under the static-axion condition \eq{axion-static},
the axion \eqref{axion-static-Tr}
becomes
\begin{align}
 \tilde T_r 
 = \frac{4 r \r^4 |A| E S}{c_e c_1} \frac{1}{|B|}
 = \sqrt{-\frac{c_2}{c_e^3}}  \frac{4r^4 \r^4  E^2 S}{c_1 |B|^{2 c_3+1}}
 \ ,
\end{align}
using the integral of motion \eq{AE-1}
and then \eq{E2-B2-A2-relation}.
Therefore,
substituting $S$ with \eqref{r2S2-eq},
we have
\begin{align}
  \tilde T_r = \r^{-2}\del_r \tilde\r
 &= \pm 4c_e\sqrt{-c_1c_2c_3}
   \frac{ r^3 |B|^{2 (1-c_3)}}{(c_3 c_e+c_1\, r^4 B^4)^2} \ \sim r^{-5}
 \ ,
 \label{axion-radial}
\end{align}
where the sign is plus if $S>0$ and minus if $S<0$.
This could become singular if 
\begin{align}
 -c_3 c_e = c_1 r^4 B^4
 ,
\end{align}
which may happen for $c_3=0$
and is associated 
with a singularity of $(\ln B)'$ (see \eq{dB-equation}); 
however, this singularity does not appear because $\tilde T_r$ is zero for $c_3=0$.

\paragraph{Dilaton.}
For $E\ne 0$, we have 
\begin{align}
 (\ln \r)'
 &= \frac 1r - \frac{1}{c_e}\frac{c_1 r^4 B^4 - c_e}{c_1 r^2 B^2} E \rho^2
  - \frac{2}{c_1 c_e} r (E \rho^2)^2
  - \frac{c_0}{2} r B E \ ,
\end{align}
using \eqref{eom-1-rho} and \eqref{AE-1}.
We then can consider the asymptotic expansion
\begin{align}
 (\ln \r)' &\sim \frac 1r - \frac{b_0^2}{c_e}  r^2 E \rho^2 
  - \frac{2}{c_1 c_e} r (E \rho^2)^2
  - \frac{c_0 b_0}{2} r E \nn\\
  &= \Big(1 - \frac{b_0^2}{c_e} (r^3 E\rho^2)\Big) \frac 1r
  - \frac{c_0 b_0}{2} (r^3E) \frac{1}{r^2}
  - \frac{2}{c_1 c_e} (r^3E)^2 \rho^4 \frac{1}{r^{5}} \nn\\
  &= -\frac{c_0 c_e}{2b_0} \r^{-2} \frac{1}{r^2}  + O(r^{-5}) 
 \ ,
\end{align}
using
\begin{align}
 \frac {b_0^2}{c_e} r^3 E \rho^2 \sim 1 + O(r^{-4})
 \ ,
\end{align}
which can be obtained from \eqref{E-asymptotics} and \eqref{rho-asymptotics}.
But this leads to
\begin{align}
 \cR = -G^{\mu\nu}({\bf T}_{\mu\nu}[\r] + {\bf T}_{\mu\nu}[\tilde\r])
 \sim r^{-4} \neq 0
 \ .
\end{align}
The $\r$ contribution clearly dominates and cannot be removed.





\paragraph{Metric and the energy-momentum tensors.}
The graphical relationship between the metric, $\tilde r$ and the traces of the energy-momentum tensors
is shown in Fig.~\ref{fig:EMT_metric_rt2_Eneq0} for $E\ne 0$
and in Fig.~\ref{fig:EMT_metric_rt2_Eeq0} for $E=0$ (with $S\ne 0$).
One can see that the traces of the energy-momentum tensor of the dilaton have poles.
Interestingly, in both cases $E\ne 0$ and $E=0$,
the scalar curvature does not diverge at 
the points where $G_{\tilde r \tilde r}$ diverges.
Moreover, we observe the scalar curvature diverges 
if $G_{tt}$ is infinity or $r$ is zero.
\begin{figure}[htbp]
 \centering
 \includegraphics[scale=0.67]{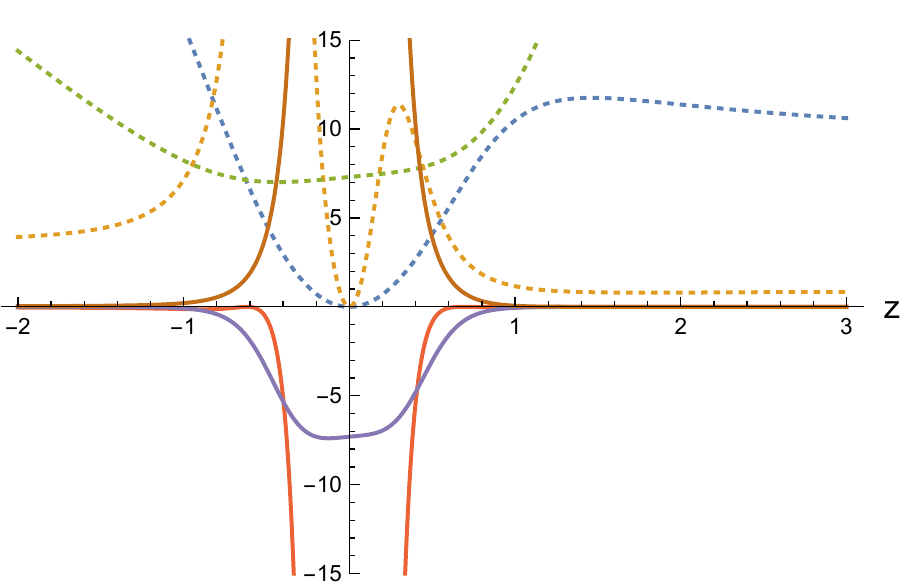}
 \includegraphics[scale=0.67]{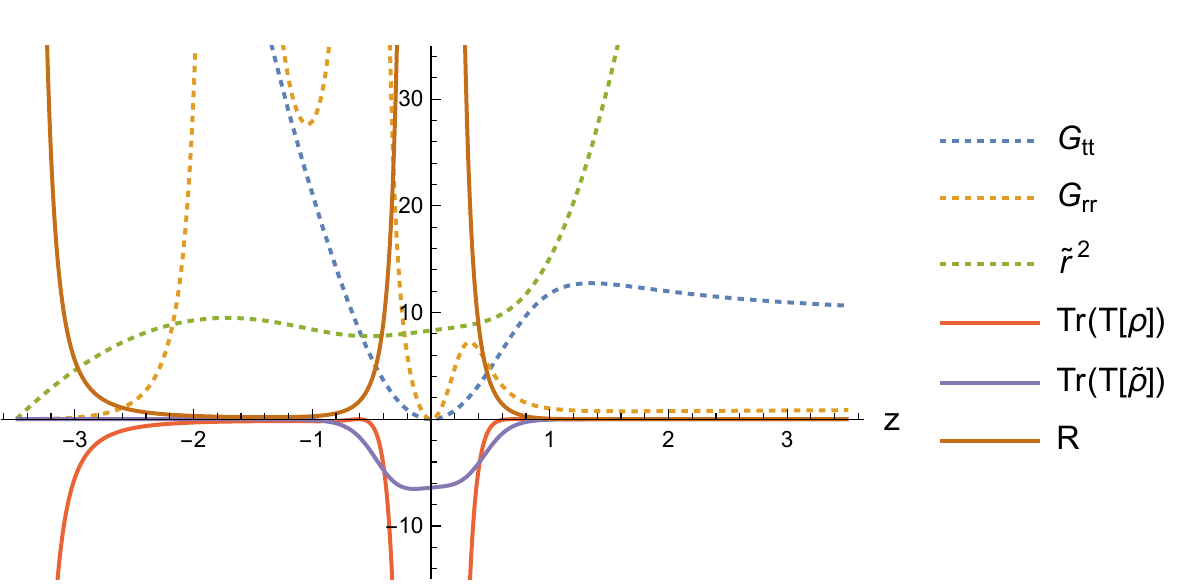}
 \includegraphics[scale=0.67]{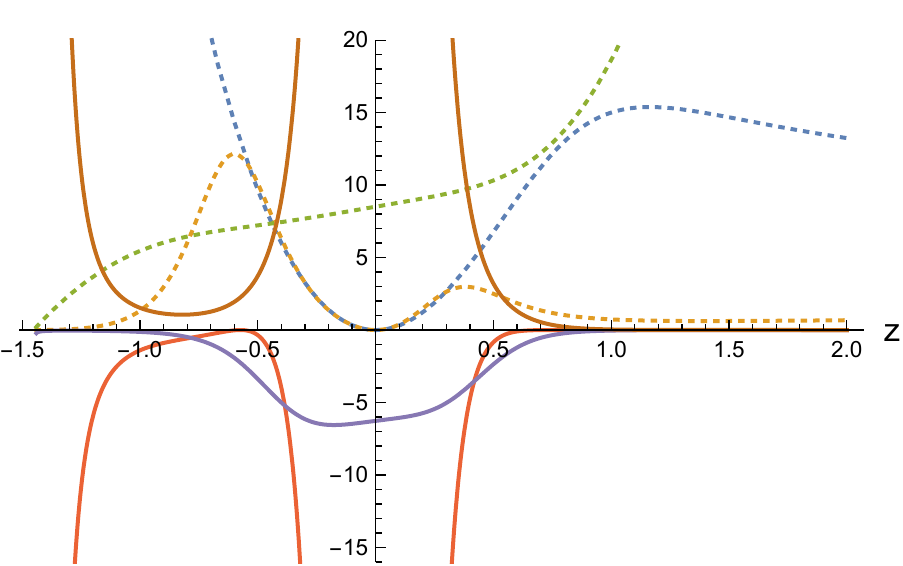}
 \includegraphics[scale=0.67]{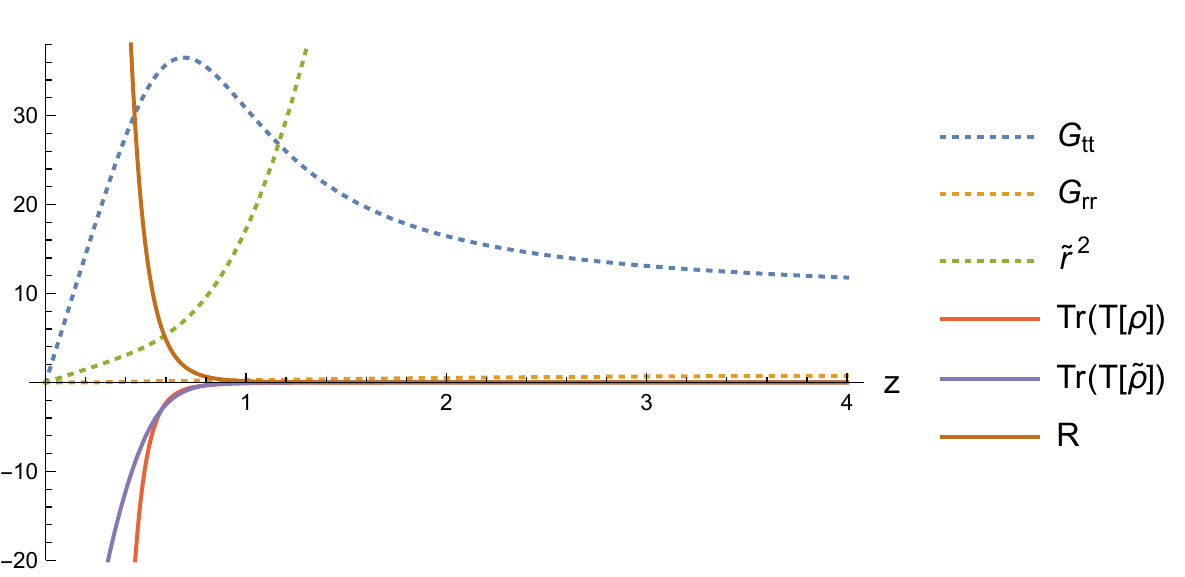}
 \includegraphics[scale=0.67]{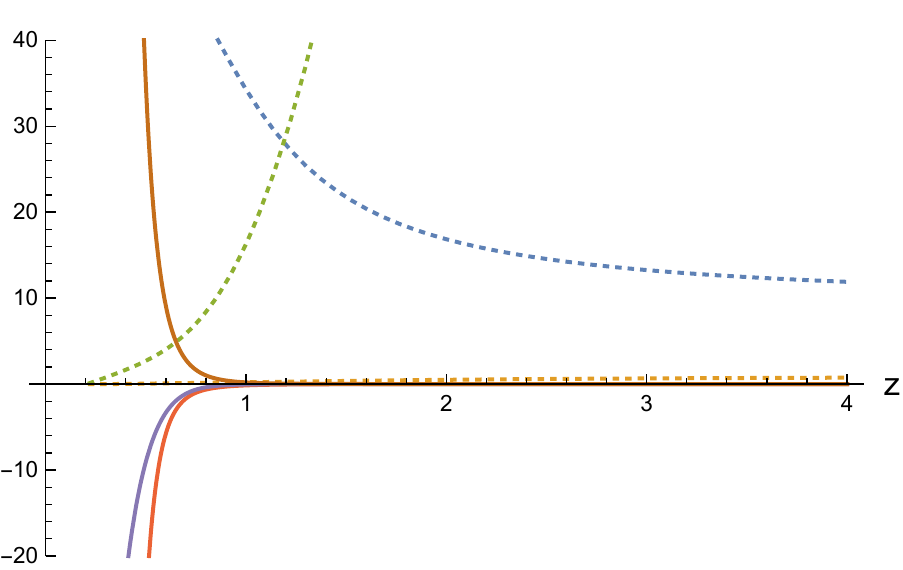}
 \includegraphics[scale=0.67]{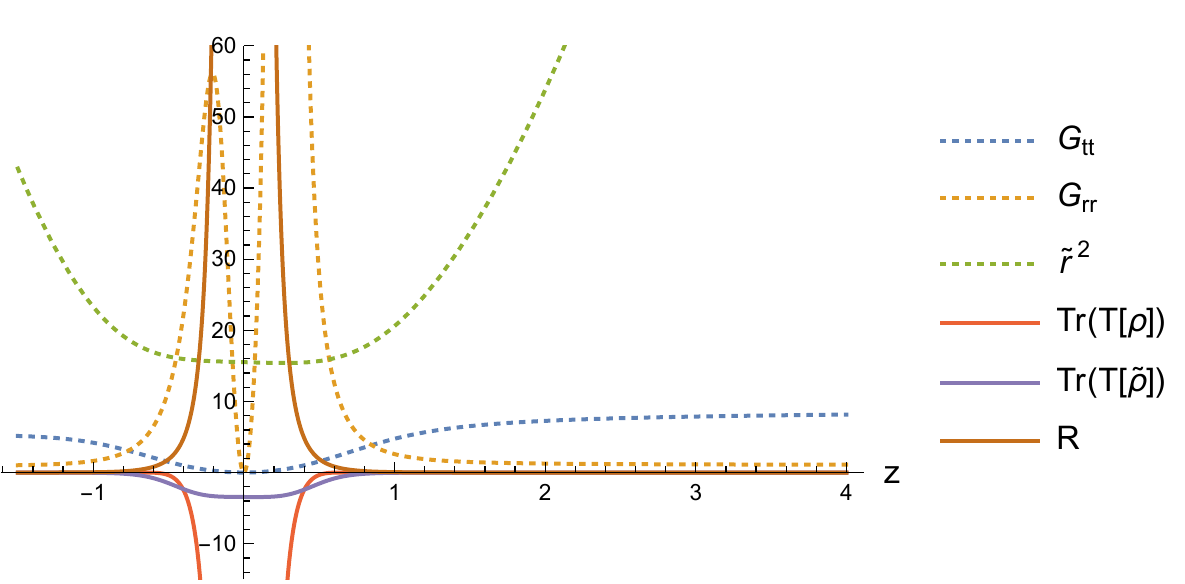}
 \caption{Graphs of the traces of the energy-momentum tensors and the scalar curvature with the metric $G_{tt}$, $G_{\tilde r\tilde r}$ and $\tilde r^2$
 for $E\neq 0$. 
 We set $c_3=0.15$, $c_0=0.4$ and $\rho_0=3$. The top-left, top-right, center-left, center-right, bottom-left and bottom-right plots are of $c_1=-6$, $-4.4$, $-2.8$, $-1.47\cdots$, $-1.4$ and $8$, respectively.
 Note that they correspond to the ones in Fig.~\ref{fig:metric_rt2_Eneq0}. In contrast to Fig.~\ref{fig:metric_rt2_Eneq0}, the metric and $\tilde r^2$ are plotted by dashed curves here.}
 \label{fig:EMT_metric_rt2_Eneq0}
\end{figure}

\begin{figure}[htbp]
 \centering
 \includegraphics[scale=0.67]{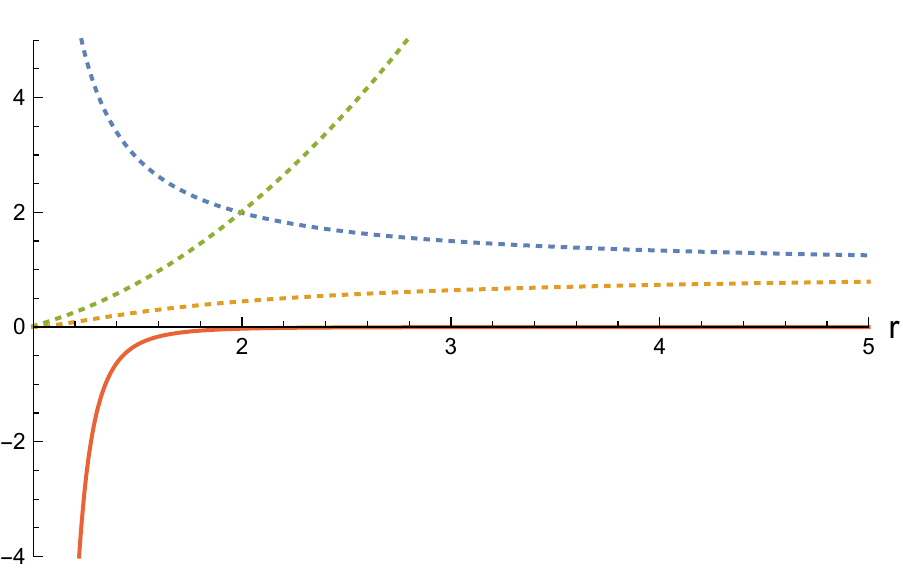}
 \includegraphics[scale=0.67]{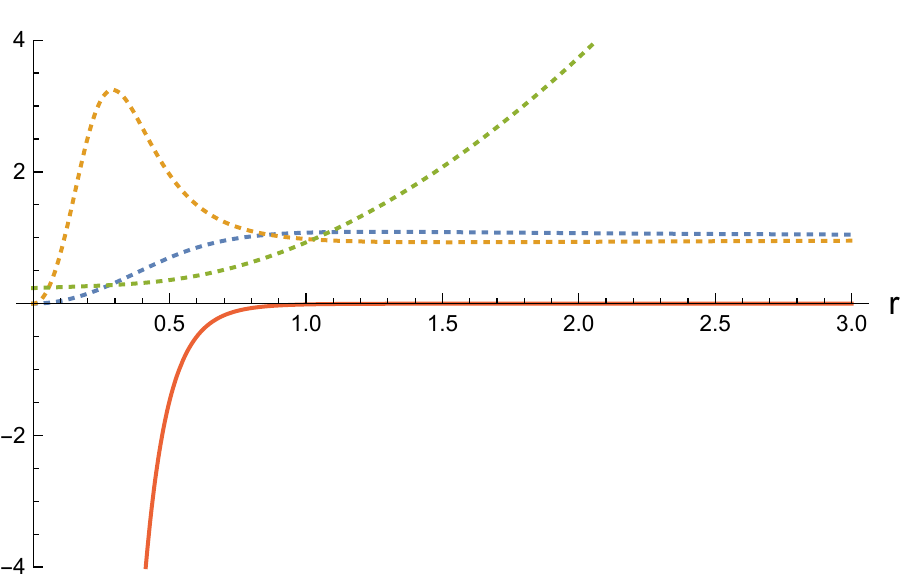}
 \includegraphics[scale=0.67]{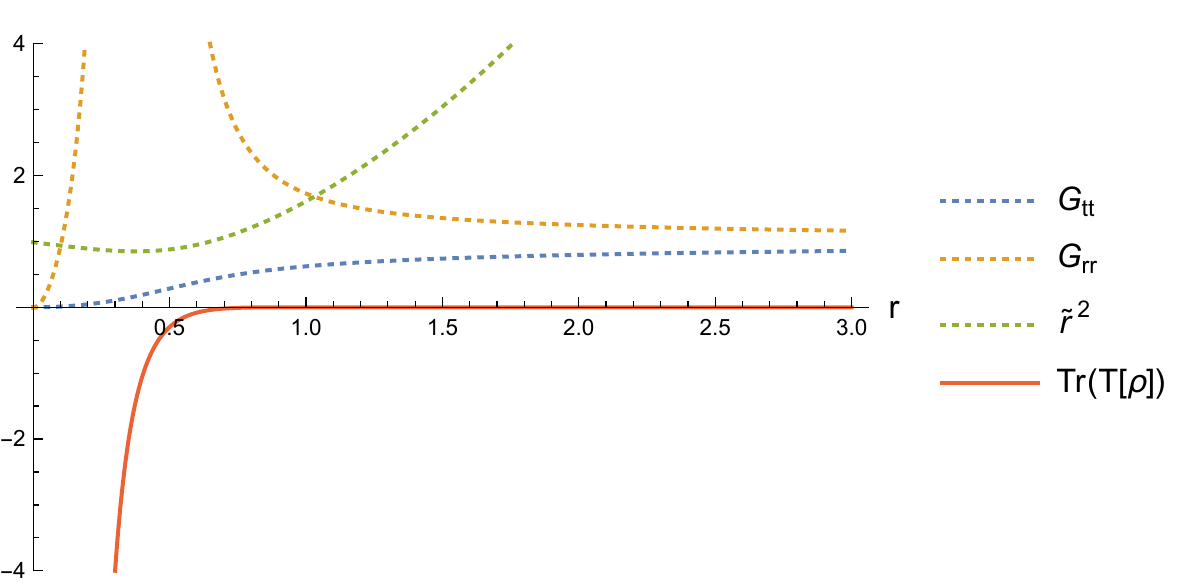}
 \caption{Graphs of the trace of the dilaton's energy-momentum tensor with the metric $G_{tt}/\tilde c_\rho$, $G_{\tilde r\tilde r}/\tilde c_\rho$ and $\tilde r^2$
 for $E= 0$. 
 We set $s_0=0.4$ and $\rho_0=3$. The top-left, top-right, bottom-left and bottom-right plots are of $c_{a1}=1$, $0.14$, $-0.5$, respectively. 
 The scalar curvature is not plotted because $\mathcal{R}=-\Tr(\mathbf{T}[\rho])$. Note also that these graphs correspond to the ones in Fig.~\ref{fig:metric_rt2_Eeq0}. Again, in contrast to Fig.~\ref{fig:metric_rt2_Eeq0}, the metric and $\tilde r^2$ are plotted by dashed curves here.} 
 \label{fig:EMT_metric_rt2_Eeq0}
\end{figure}

\section{Discussion and outlook}
We derived and classified 
general $SO(3)$-symmetric static solutions to the matrix-model equations of motion 
written in terms of the frame of $\mathcal{M}^{3,1}$.
We identify two different classes of solutions: one without axion for $E=0$ 
in our ansatz, and one including an axion for $E\neq 0$.
We imposed a natural condition \eqref{axion-static}, 
in which the axion is also static.
The behavior of the solutions has several remarkable features.

Firstly, the asymptotic behavior of the general solutions is 
the same as the previously found special solution \eqref{special-solution-HS},
for sufficiently large $r$.
This suggests that the special solution well represents 
the general solutions to the equations of motion, which include further contributions 
due to the axion and dilaton at shorter distances.

Secondly in the $E\ne 0$ case, 
it is remarkable that there is a parameter region 
where the effective metric is double-valued in the effective radius squared $\tilde r^2$.
More precisely, in such a solution
there are two values of $z$ corresponding to a value of $\tilde r^2$,
so that there are two corresponding values of the other elements of the metric.
The metric has a singularity at the minimal $\tilde r^2$.
This can be viewed as a wormhole-like solution.
Since there is no symmetry of flipping $z$ around the minimum of $\tilde r^2$,
it turned out the behaviors of the two sides of the metric are entirely different in general.


Finally, we found that
the Schwarzschild solution is not contained in the general solutions,
though the asymptotic behavior of the general metric reproduces the 
linearized Schwarzschild metric.
Although this sounds like an unwanted result,
we can expect to recover such a solution upon taking into account the Einstein-Hilbert term, which is 
induced by one-loop effects \cite{Steinacker:2021yxt}.
In particular, it is interesting that the origin of the (asymptotic) mass of the 
present solutions is not some singular matter localized at the center, 
but a new type of ``vacuum energy'' due to dilaton, axion, and 
the noncommutative frame itself.
Therefore the solutions  obtained in this paper should be understood
as solutions to the pre-gravity theory arising on classical brane solutions, and their full significance
will only be understood upon taking into account the induced Einstein-Hilbert term.
The results and techniques developed here should allow 
to study also  solutions of such a combined action. This is certainly the most important open problem,
which is postponed for future work.

There are clearly many further directions for follow-up work.
For example, we studied only static solutions in this paper.
Since the classical solutions obtained without the Einstein-Hilbert term
are expected to describe spacetime in the cosmological regime,
it will be intriguing to investigate more general solutions with cosmic time evolution.
Establishing a 
relation of such cosmic solutions to numerical attempts to realise cosmic time evolution from the matrix model, e.g.~\cite{Kim:2011cr,Nishimura:2019qal},
would significantly improve our understanding of the matrix model and the universe.
On a more technical level, we have provided a partial answer to the problem of reconstructing matrix configurations 
for a given frame.
A more complete treatment of this problem and of possible higher-spin contributions should be given elsewhere.

\section*{Acknowledgment}
Y.A.~thanks Katsuta Sakai for useful discussions.
This work was supported by the Austrian Science Fund (FWF) grant P32086.
The Mathematica package Riemannian Geometry \& Tensor Calculus (RGTC), coded by Sotirios Bonanos,
was of great help in checking the solutions.

\appendix

\section{Reconstruction of the frame}

In this paper, we have solved the equations of motion for the frame 
$\tensor{E}{_{\dot\a}^\mu}$. However, it remains to be shown that 
such frames can indeed arise as configurations in the semi-classical matrix model, i.e.~via
Poisson brackets \eq{frame-Poisson}
\begin{align}
  \{Z_{\dot\a},x^\mu\} = \tensor{E}{_{\dot\a}^\mu} \qquad  \in\cC^0
  \label{frame-reconstruction}
\end{align}
in terms of the basic matrices $Z_{\dot\a}$. For the solution found in \cite{Fredenhagen:2021bnw},
that question was settled directly by constructing suitable $Z_{\dot\a}$.
This was possible, but required an infinite tower of higher-spin contributions to $Z_{\dot\a}$.

Here we  want to address this  question more generally: Given any (spin 0 valued)
frame $\tensor{E}{_{\dot\a}^\mu}\in\cC^0$ 
which satisfies the divergence constraint  $\del_\nu\big(\r_M\tensor{E}{^\nu}\big) = 0$ 
\eq{frame-div-free}, are there always generators (``potentials'') $Z_{\dot\a}$
such that \eq{frame-reconstruction} holds?

\vspace{0.2cm}

We can provide a partial answer to this question: 
For any divergence-free frame $\tensor{E}{_{\dot\a}^\mu}\in\cC^0$, we can construct 
generators $Z_{\dot\a} \in \cC^1$ such that  $ [\{Z_{\dot\a},x^\mu\}]_0 = \tensor{E}{_{\dot\a}^\mu}$ for the projection $[.]_0$ on $\cC^0$.
However, we cannot settle the question if this equation can  be satisfied 
for all higher-spin components for suitable $Z_{\dot\a}$;
this  is postponed to future work.
Moreover we restrict ourselves to the asymptotic regime, where the wavelengths 
are much smaller than the cosmic scale. 
To show this (partial) result, we first establish some results for the 
fuzzy hyperboloid $H^4_n$ which is underlying $\cM^{3,1}$.

\subsection{Divergence-free vector fields on $H^4$ and reconstruction}
\label{sec:reconstruction}

First we recall some results for the fuzzy hyperboloid  $H_n^4 \subset \R^{4,1}$
(see (4.9) and (9.9) in \cite{Steinacker:2019awe}):
\begin{lem}
Given a divergence-free tangential vector field $\phi_{a}\in\cC^0$ on $H^4$,
we have
\begin{align}
 - \{x_a,\phi^{(1)}\}_{0} &=  \frac{1}{3}  (\Box_H - 2 r^2) \phi_a \ , \qquad \mbox{where} \ \
 \phi^{(1)} = \{x^a,\phi_a\} \ 
 \label{x-phi-comm-spins}
\end{align}
where $x^a$ are  Cartesian coordinates on  $\R^{4,1}$.
 \label{lemma-alpha}
 \end{lem}

\begin{lem}
\label{lemma:box-intertwine-VF}
\begin{align}
  (\Box_H  + 2 r^2 s)\{\phi^{(s)},x_a\}_- &= \{\Box_H\phi^{(s)},x_a\}_-
 \label{Box-H-intertwiner-minus} \\
 (\Box_H  - 2 r^2 (s+1))\{\phi^{(s)},x_a\}_+ &=\{\Box_H\phi^{(s)},x_a\}_+ \ 
 \label{Box-H-intertwiner-plus}
\end{align}
for any $\phi^{(s)} \in \cC^s$. Here $[]_\pm$ denotes the projection on $\cC^{s\pm 1}$.
 
\end{lem}

Now consider the following (possibly $\hs$-valued)  vector fields on $H^4$
\begin{align}
 V:= V^a \eth_a, \qquad V^a = \{Z,x^a\}
\end{align}
generated by some $Z \in\cC$, where $\eth_a$ is the tangential derivative operator on $H^4_n$ introduced in \cite{Sperling:2019xar}.
It was shown there that
$V$ is always tangential to $H^4$ and  divergence-free,
\begin{align}
 \eth_a V^a = 0 \ .
 \label{VF-H-divfree}
\end{align}
$V$ can be viewed as push-forward of the Hamiltonian vector field
$\{Z,.\}$ on $\C P^{1,2}$
to $H^4$ via the bundle projection\footnote{In general, the push-forward of a vector
field via a non-injective map is not well defined. However, the push-forward in the present situation makes sense if interpreted as $\hs$-valued map.}.
For $Z=Z^{(s)} \in \cC^s$, it decomposes into different $\hs$ components as
\begin{align}
 V^a = V^a_{(s-1)} + V^a_{(s+1)} \qquad \cC^{s-1} \oplus \cC^{s+1} \ .
\end{align}
One might hope that all divergence-free $\hs$-valued vector fields can be written in this form, but this is not possible,
by counting degrees of freedom. However for $s=0$,
all divergence-free vector fields $V^a\in\cC^0$
can indeed be obtained in this way for a suitable $Z \in \cC^1$, up to $\hs$ corrections. 
More precisely, we have the following  result:

\begin{lem}
 \label{lem:H4-reconstruct}
Given any divergence-free tangential vector field $\eth_a V^a = 0$ on $H^4$
with $V^a \in \cC^0$,
there is a unique generator $Z\in\cC^1$ such that 
\begin{align}
 V^a = \{Z,x^a\}_0 \ .
 \label{VF-reconstruction-H4}
\end{align}
This $Z$ is given explicitly by
\begin{align}
 Z := -3(\Box_H -4r^2)^{-1}\{V^a,x_a\} \qquad \in \cC^1 \ .
\end{align}
 
\end{lem}

\begin{proof}
 

As pointed out above, the vector field $\{Z,x^a\}$ is always divergence-free,
\begin{align}
\eth_a \{Z,x^a\} = 0 \ .
\end{align}
The intertwiner result \eq{Box-H-intertwiner-minus} implies 
\begin{align}
  (\Box_H -2r^2) \{Z^{(1)},x_a\}_- &= \{((\Box_H-4r^2)Z^{(1)},x_a\}_- \nn\\
   \{(\Box_H-4r^2)^{-1}Z^{(1)},x_a\}_- &= (\Box_H-2r^2)^{-1}\{Z^{(1)},x_a\}_- 
 \label{box-inv-int} 
\end{align}
for any $Z^{(1)} \in \cC^1$.
Therefore
\begin{align}
 \{Z,x^a\}_0 &= -3\{(\Box_H-4r^2)^{-1}\{V^b,x_b\},x^a\}_0 
  = -3(\Box_H-2r^2)^{-1} \{ \{V^b,x_b\},x^a\}_0 = V^a
\end{align}
using Lemma \ref{lemma-alpha} in the last step.
Uniqueness can be seen similarly using the results in \cite{Steinacker:2019awe}.
The inverse operators make sense at least for square-integrable functions, because 
 $\Box_H-\frac 92r^2$ is  positive-definite on $H^4$ for unitary irreps (cf.~\cite{Steinacker:2019awe}),
 and $\Box_H-4r^2$ is positive-definite on unitary modes in $\cC^1$. 

\end{proof}

Now consider the reconstruction problem on $H^4_n$.
Given $V^a$, we define 
\begin{align}
 Z^{(1)} := -3(\Box_H -4r^2)^{-1}\{V^a,x_a\} \qquad \in \cC^1 \ 
\end{align}
which satisfies 
\begin{align}
 V^a = \{Z^{(1)},x^a\}_0 \ 
\end{align}
as shown above.
However, $\{Z^{(1)},x^a\}$ contains  in general also a spin 2 component
\begin{align}
 V^{(2)a} := \{Z^{(1)},x^a\}_2 \qquad \in \cC^2 \ .
 \label{vectorfield-C2-component}
\end{align}
Noting that $\eth$ respects $\cC^n$, this satisfies 
\begin{align}
  \eth_a V^{(2)a} &= 0,  \nn\\
  \{V^{(2)a},x_a\}_3 &= \{\{Z^{(1)},x^a\}_2,x_a\}_3 = 0 \ 
  \label{V-2-ids}
\end{align}
since $\cC^3$ does not contain any spin 1 mode.
The second relation means that we cannot just repeat the above procedure to cancel this. 
One can show that the only generators $Z\in\cC^1$ which do 
{\em not} induce spin 2 components via \eq{vectorfield-C2-component}
are linear combinations of $\theta^{ab}$.

This means that the reconstruction of vector fields on  
$H^4$ generically leads to extra higher-spin components  $V^{(2)a} \in \cC^2$ \eq{vectorfield-C2-component}, which however encode the same information as $V^a$.
It remains an open question if these can be cancelled by suitable $\hs$-modifications of $Z$ and possibly $x^a$.

\subsection{Divergence-free vector fields and reconstruction on $\cM^{3,1}$}
\label{sec:hs-valued-VF-M31}

Now we recall that $\cM^{3,1}\subset \R^{3,1}$ is obtained from $H^4\subset \R^{4,1}$ via a projection along $x^4$.
Therefore any vector field $V^a$  on $H^4$ can be projected to a vector field $V^\mu$ on $\cM^{3,1}$ by
simply dropping the $V^4$ component. In particular, the Hamiltonian vector field $V^a = \{Z,x^a\}$
is mapped to $V^\mu = \{Z,x^\mu\}$ in Cartesian coordinates.
Conversely, any vector field $V^\mu$ on $\cM^{3,1}$ can be lifted to $H^4$ by defining 
\begin{align}
V^4 := - \frac{1}{x_4} x_\mu V^\mu \ ,
\label{V4-def}
\end{align}
which clearly satisfies the tangential relation $V^a x_a = 0$ on $H^4$.

It turns out that 
this correspondence maps divergence-free ($\cC^0$-valued) vector fields $\eth_a V^a = 0$ on $H^4$
to divergence-free vector fields  on $\cM^{3,1}$ and vice versa, in the sense that
\begin{align}
 \del_\mu (\r_M V^\mu) = 0, \qquad \r_M = \sinh(\eta)^{-1} \ .
\end{align}
This is established in the following result:

\begin{lem}
 \label{lemma-div-M31}

Let $V^a$ be a ($\cC^0$-valued, tangential) divergence-free vector field on $H^4$.
Then its reduction (or push-forward) to $\cM^{3,1}$ satisfies
\begin{align}
 \del_\mu(\r_M V^\mu) =0 \ .
 \label{div-free-M31}
\end{align}
Conversely, if $V^\mu$ satisfies \eq{div-free-M31}
then its lift to $H^4$ defined by \eq{V4-def}
is divergence-free in the sense of \eq{VF-H-divfree}.

\end{lem}

\begin{proof}

First we recall the following property of the tangential derivative $\eth$ from  (3.65) in \cite{Steinacker:2019awe}
\begin{align}
 \eth^\mu x^\nu &= \eta^{\mu\nu} + \frac 1{R^2} x^\mu x^\nu\,
  = \big(\del^\mu + \frac 1{R^2} x^\mu x^\s\del_\s\big) x^\nu \ .
 \label{}
\end{align}
Therefore we can identify 
\begin{align}
 \eth^\mu = \del^\mu + \frac 1{R^2} x^\mu x^\s\del_\s  \qquad \mbox{on} \ \ \cC^0 \ .
 \label{eth-explicit}
\end{align}
Furthermore, 
\begin{align}
 \eth^4 Z &= \frac{1}{r^2 R^2} x_\mu \{\theta^{4\mu} ,Z\}
  = \frac{1}{R} x^\mu\{t_\mu,Z\}
  =\frac{\sinh(\eta)}{R} x^\mu \del_\mu Z   
  \label{eth4-explicit}
\end{align}
for $Z\in\cC^0$,
since $\theta^{4\mu}= r^2\cM^{\mu 4} = r^2 R t^\mu$.
Therefore
\begin{align}
 \eth_a V^a &=  \eth_\mu V^\mu  + \eth_4 V^4   \nn\\
  &= \big(\del_\mu + \frac 1{R^2} x_\mu x^\s\del_\s\big)  V^\mu 
  - \frac{\sinh(\eta)}{R^2} x^\nu\del_\nu(\sinh(\eta)^{-1} x_\mu V^\mu)  \nn\\
  &= \sinh(\eta) \del_\mu(\sinh(\eta)^{-1} V^\mu) \ .
\end{align}

\end{proof}

\paragraph{Reconstruction of vector fields and frame.}

We can now solve the following ``reconstruction'' problem on $\cM^{3,1}$:
Given any $\cC^0$-valued divergence-free vector field $V^\mu$,
\begin{align}
 \del_\mu(\r_M V^\mu) = 0 \ ,
\end{align}
there is a generating function $Z \in\cC^1$ such that 
\begin{align}
 V^\mu = \{Z,x^\mu\}_0 \ .
 \label{Z-generate-V}
\end{align}
This can be obtained by lifting $V^\mu$ to a divergence-free vector field 
$V^a$ on $H^4$
as in Lemma \ref{lemma-div-M31}. Then the  result \eq{VF-reconstruction-H4} on $H^4$
states that $V^a = \{Z,x^a\}_0$ for some $Z\in\cC^1$, 
which implies $V^\mu = \{Z,x^\mu\}_0$.
Explicitly, this $Z$ is given by
\begin{align}
 Z &= -3(\Box_H -4r^2)^{-1}\big(\{V^\mu,x_\mu\} + \{V^4,x_4\}   \big)  \nn\\
  &=  -3(\Box_H -4r^2)^{-1}\big(\{V^\mu,x_\mu\} 
  + \frac{1}{x_4}\{x_4, x_\mu V^\mu\}   \big) \ .
\end{align}
In particular, for given any classical frame $\tensor{E}{_\a^\mu}$ there is a 
unique $Z_\a \in \cC^1$ such that $\tensor{E}{_\a^\mu} = \{Z_\a,x^\mu\}_0$.
E.g.~for the cosmic background frame, this gives $E_{\a 4} = 0$, and 
we recover
\begin{align}
Z_\a 
 = -12(\Box_H-4r^2)^{-1}\{\sinh(\eta),x_\a\} = t_\a \ .
\end{align}
The generator $Z_\a\in\cC^1$ is uniquely determined by \eq{Z-generate-V}.
This means that the corresponding spin 2 part $\{Z,x^\mu\}_+ \in \cC^2$ 
is also uniquely determined by the vector field.
Therefore in general, the reconstructed frame will contain  higher spin components.
These higher-spin components  cancel  upon averaging over $S^2_n$
 in the linearized theory, but not in the non-linear regime. 
 Since the $\hs$ components of the generators $Z_\a$ in $\cC^{s}$ for $s\geq 2$
are undetermined, it is plausible that these can be adjusted such that
the unwanted higher-spin components of the frame cancel (possibly upon redefining  $x^\mu$),
as in the rotationally invariant solution in \cite{Fredenhagen:2021bnw}.

In any case, these $\hs$ components encode the same vector field as the underlying 
spin $0$ component of the frame, cf.~(D.18) in \cite{Sperling:2018xrm}, since both are encoded in 
$Z_\a \in \cC^1$. 
This means that in a contraction of the frame (such as the metric) 
or of the torsion (such as in the Einstein-Hilbert action), 
the averaged contribution of these $\hs$ components over the internal $S^2_n$
should be similar to the spin 0 component; however this needs to be established in detail elsewhere.
Once this is understood, one may also try to relate our 
solutions with analogous solutions \cite{Iazeolla:2017vng,Iazeolla:2011cb}
obtained in Vasiliev higher spin theory, notably after taking into account the induced 
Einstein-Hilbert term \cite{Steinacker:2021yxt}.
All this remains to be studied in more detail elsewhere.

\section{Deriving the equation of motion}\label{sec:derive-EoM}
Let us briefly review how to derive the equation of motion~\eqref{torsion-frame-eom-3}.

We start from the bosonic part of the matrix-model action \eqref{MM-action}:
\begin{align}
 S_\text{bos}[Z]=
 \Tr\left(
 [Z^{\dot\alpha}, Z^{\dot\beta}][Z_{\dot\alpha}, Z_{\dot\beta}]
 +2m^2Z_{\dot\alpha}Z^{\dot\alpha}
 \right)
 .
\end{align}
The equation of motion of the matrix model is
\begin{align}
 [Z^{\dot\beta}, [Z_{\dot\beta}, Z_{\dot\alpha}]]
 -m^2Z_{\dot\alpha}
 =0,
 \label{eom_matrix}
\end{align}
which reduces to
\begin{align}
 \{ Z^{\dot\beta}, \{ Z_{\dot\beta}, Z_{\dot\alpha} \}\}
 +m^2 Z_{\dot\alpha}
 =0,
 \label{eom_Poisson}
\end{align}
in the semi-classical limit, 
where the endomorphism algebra $\End(\cH)$ becomes a commutative algebra of functions.

This equation of motion can be rewritten 
via the frame by a simple computation.
Let us first denote the Hamiltonian vector field for a field $M$ on the manifold 
where the functions reside in the semi-classical limit,
by
\begin{align}
 \mathcal{P}(M)f :=i\{ M, f \}.
\end{align}
This satisfies
\begin{align}
 [\mathcal{P}(M_1), \mathcal{P}(M_2)] = \mathcal{P}(i\{ M_1, M_2 \}),
\end{align}
and therefore for $Z_{\dot\alpha}$, 
\begin{align}
 [\mathcal{P}(Z^{\dot\beta}), [\mathcal{P}(Z_{\dot\beta}), \mathcal{P}(Z_{\dot\alpha})]]
 =\mathcal{P}(-\{ Z^{\dot\beta}, \{ Z_{\dot\beta}, Z_{\dot\alpha}\}\}),
\end{align}
where $[\ast , \ast\, ]$ is the usual commutator.
Thus the equation of motion \eqref{eom_Poisson}
can be computed by
\begin{align}
 [\mathcal{P}(Z^{\dot\beta}), [\mathcal{P}(Z_{\dot\beta}), \mathcal{P}(Z_{\dot\alpha})]]
 =m^2\mathcal{P}(Z_{\dot\alpha})
 .
\end{align}
The action of $Z_{\dot\alpha}$ on a function
can be written by
a Weitzenb\"ock connection $\nabla_{\dot\alpha}$
if $\tensor{E}{_{\dot\alpha}^{\mu}}$ is 
a set of globally defined linear independent frame fields, 
i.e.~$\mathcal{M}^{3,1}$ is parallelizable.
The relation between them is $\mathcal{P}(Z_{\dot\alpha})=i\nabla_{\dot\alpha}$
since
\begin{align}
 \mathcal{P}(Z_{\dot\alpha}) f
 =i\{ Z_{\dot\alpha}, f \}
 =i\tensor{E}{_{\dot\alpha}^{M}}\partial_M f
 =i\nabla_{\dot\alpha} f.
\end{align}
It acts 
as $\nabla_{\dot\alpha}=\tensor{E}{_{\dot\alpha}^\mu}\partial_\mu$
without a spin-connection term
on fields without any general coordinate indices
while it acts on contravariant vector fields as
$\nabla_\mu V^\nu=\partial_\mu V^\nu +\tensor{\Gamma}{_{\mu\rho}^\nu}V^\rho$.

Then the equation of motion is written as a relation for operators\footnote{
The mathematical structure of the equations for the frame
is essentially the same as Hanada-Kawai-Kimura
\cite{Hanada:2005vr}.
However in that approach, the matrices are interpreted as
differential operators on a commutative bundle over space-time
(see e.g.~Ref.~\cite{Hanada:2006gg,Furuta_2007,Isono:2009pu,Asano:2012mn,Sakai:2019cmj} for details),
while here they are quantized functions on a bundle over space-time.
Accordingly, the space of modes in $\End(\cH)$ is vastly bigger in 
Hanada-Kawai-Kimura, and the absence of ghosts has not been established.
Moreover, the covariant derivative here is the one with the Weitzenb\"ock connection
while, in Hanada-Kawai-Kimura, it is the standard Levi-Civita connection
multiplied by a Clebsch-Gordan coefficient for the decomposition of 
the tensor product of a vector representation and a regular representation 
into regular representations.
}:
\begin{align}
 [\nabla^{\dot\beta},[\nabla_{\dot\beta}, \nabla_{\dot\alpha}]]
 =-m^2\nabla_{\dot\alpha}
 \ .
\end{align}
Since the commutator of the covariant derivatives satisfies
\begin{align}
 &[\nabla_{\dot\alpha}, \nabla_{\dot\beta}]
 =-\tensor{T}{_{\dot\alpha\dot\beta}^{\dot\gamma}}\nabla_{\dot\gamma},
\end{align}
where
$\tensor{T}{_{\dot\alpha\dot\beta}^{\dot\gamma}}=\tensor{E}{_{\dot\alpha}^\mu}\tensor{E}{_{\dot\beta}^\nu}\tensor{T}{_{\mu\nu}^{\dot\gamma}}$ is the torsion for the Weitzenb\"ock connection,
the equation of motion becomes\footnote{
This form of the equation of motion in this paper is
the same as Ref.~\cite{Furuta_2007}
because the Riemann curvature with the contribution from the torsion 
is zero in the Weitzenb\"ock connection.
}
\begin{align*}
 \nabla_{\dot\beta}\tensor{T}{^{\dot\beta}_{\dot\alpha}^{\dot\gamma}}
 +\tensor{T}{_{\dot\alpha\dot\beta}^{\dot\delta}}
 \tensor{T}{^{\dot\beta}_{\dot\delta}^{\dot\gamma}}
 -m^2\delta_{\dot\alpha}^{\dot\gamma}
 =0,
\end{align*}
or partially in terms of the general coordinate indices,
\begin{align}
 \nabla_{\mu}(\gamma^{\mu\nu}\tensor{T}{_{\nu\rho}^{\dot\alpha}})
 +\tensor{T}{_{\rho\mu}_{\nu}}
 \tensor{T}{^{\mu\nu}^{\dot\alpha}}
 -m^2\tensor{E}{^{\dot\alpha}_{\rho}}
 =0.
\end{align}
This is the equation derived in \cite{Steinacker:2020xph}, where
$T^{\mu\nu\dot\alpha}=\gamma^{\mu\mu'}\gamma^{\nu\nu'}\tensor{T}{_{\mu'\nu'}^{\dot\alpha}}$.

Let us then rewrite the above equation of motion
in terms of the Levi-Civita connection.
The relation of the Weitzenb\"ock connection with 
the Levi-Civita connection is
\begin{align}
 \tensor{\Gamma}{_{\mu\nu}^{\rho}}
 &=-\tensor{E}{^{\dot\alpha}_\nu}\partial_{\mu}\tensor{E}{_{\dot\alpha}^{\rho}}
 =\tensor{\Gamma}{^{(\gamma)}_{\mu\nu}^{\rho}}
 +\tensor{K}{_{\mu\nu}^{\rho}}
 \nonumber \\
 &=\tensor{\Gamma}{^{(G)}_{\mu\nu}^{\rho}}
 +\tensor{K}{_{\mu\nu}^{\rho}}
 -\rho^{-1}\left(
 \delta_{\nu}^{\rho}\partial_\mu\rho
 +\delta_{\mu}^{\rho}\partial_\nu\rho
 -\gamma_{\mu\nu}\gamma^{\sigma\rho}\partial_\sigma\rho
 \right)
 ,
 \label{Weitzenboeck2Gandrho}
\end{align}
where $\tensor{\Gamma}{^{(\gamma)}_{\mu\nu}^{\rho}}$ and 
$\tensor{\Gamma}{^{(G)}_{\mu\nu}^{\rho}}$ are
the Levi-Civita connections associated with $\gamma_{\mu\nu}$ and $G_{\mu\nu}$,
respectively, and
$\tensor{K}{_{\mu\nu}^{\rho}}$
is the contorsion of the Weitzenb\"ock connection,
which is also interpreted as the spin connection
constructed from $\tensor{\Gamma}{^{(\gamma)}_{\mu\nu}^{\rho}}$
via
$\tensor{K}{_{\mu}^{\dot\alpha\dot\beta}}=-\tensor{E}{^{\dot\beta}_\rho}\nabla^{(\gamma)}_{\mu}\tensor{E}{^{\dot\alpha\rho}}$.
We denote 
the covariant derivative associated with $\gamma_{\mu\nu}$ and $G_{\mu\nu}$
by $\nabla^{(\gamma)}_\mu$ and $\nabla^{(G)}_\mu$, respectively.
Substituting the Weitzenb\"ock connection with \eqref{Weitzenboeck2Gandrho}
and using the relation 
$\tensor{T}{_{\mu\nu}^{\rho}}=\tensor{K}{_{\mu\nu}^{\rho}}-\tensor{K}{_{\nu\mu}^{\rho}}$,
one obtains
\begin{align}
 &\nabla^{(G)}_{\mu}(\gamma^{\mu\nu}\tensor{T}{_{\nu}_{\rho}^{\dot\alpha}})
 +\frac{1}{2}(\tensor{T}{_{\mu\nu\rho}}
 +\tensor{T}{_{\nu\rho\mu}}
 +\tensor{T}{_{\rho\mu\nu}})
 \tensor{T}{^{\mu\nu\dot\alpha}}
 \nonumber \\
 &\hspace{80pt}
 +\gamma^{\mu\nu}\tensor{T}{_{\sigma\mu}^\sigma}\tensor{T}{_{\nu}_{\rho}^{\dot\alpha}}
 -(\delta_{\sigma}^{\sigma}-2)
 \gamma^{\mu\nu}
 \rho^{-1} \partial_\mu\rho\, 
 \tensor{T}{_{\nu}_{\rho}^{\dot\alpha}}
 -m^2\tensor{E}{^{\dot\alpha}_{\rho}}
 =0.
 \label{EOM_before_div-free}
\end{align}

The Jacobi identity for $\theta^{\mu\nu}$,
which reduces to the identity $\partial_\nu(\rho_M\theta^{\mu\nu})=0$,
results in the divergence constraint \eqref{frame-div-free}.
Namely, 
one can show that the frame satisfies
\cite{Steinacker:2020xph}
\begin{align*}
 \partial_\mu(\sqrt{|G|}\rho^{-2}\tensor{E}{_{\dot\alpha}^\mu})
 =\partial_\mu(\rho_M\tensor{E}{_{\dot\alpha}^\mu})
 =-\partial_\mu(\rho_M\theta^{\mu\nu}\partial_\nu Z_{\dot\alpha})
 =-\partial_\mu(\rho_M\theta^{\mu\nu})\partial_\nu Z_{\dot\alpha}
 =0
 ,
\end{align*}
and hence $\tensor{E}{^{\dot\alpha}_\mu}\partial_\nu \tensor{E}{_{\dot\alpha}^\nu}=-\partial_\mu\ln\rho_M$.
Therefore, \eqref{T-trace-rho} holds:
\begin{align*}
 \tensor{K}{_{\sigma\mu}^\sigma}
 =\tensor{T}{_{\sigma\mu}^\sigma}
 =-\tensor{E}{^{\dot\alpha}_\mu}\partial_\nu \tensor{E}{_{\dot\alpha}^\nu}
 +\tensor{E}{^{\dot\alpha}_\nu}\partial_\mu \tensor{E}{_{\dot\alpha}^\nu}
 =\partial_\mu \ln\left[
 \frac{\rho_M}{\sqrt{|\gamma|}}
 \right]
 =\frac{2}{\rho}\partial_\mu \rho
 .
\end{align*}
By plugging this equation 
into the equation of motion \eqref{EOM_before_div-free},
one reaches
\begin{align}
 \nabla^{(G)}_{\mu}(\gamma^{\mu\nu}\tensor{T}{_{\nu}_{\rho}^{\dot\alpha}})
 +\frac{1}{2}(\tensor{T}{_{\mu\nu\rho}}
 +\tensor{T}{_{\nu\rho\mu}}
 +\tensor{T}{_{\rho\mu\nu}})
 \tensor{T}{^{\mu\nu\dot\alpha}}
 -m^2\tensor{E}{^{\dot\alpha}_{\rho}}
 =0.
\end{align}
Finally, the following relation 
\begin{align}
 \rho^2 \sqrt{|G|}^{-1}
 \varepsilon^{\mu\nu\rho'\sigma}G_{\r\r'} \tilde T_{\sigma} \,
 =-\gamma^{\mu\mu'}\gamma^{\nu\nu'}
 (\tensor{T}{_{\mu'\nu'\rho}}
 +\tensor{T}{_{\nu'\rho\mu'}}
 +\tensor{T}{_{\rho\mu'\nu'}})
\end{align}
leads us to the form of the equation of motion \eqref{torsion-frame-eom-3}.

\bibliographystyle{JHEP}
\bibliography{papers}


\end{document}